\documentclass{article}

\PassOptionsToPackage{numbers, compress}{natbib}
\usepackage[final]{neurips_2018}

\usepackage{microtype}
\usepackage{booktabs} % for professional tables
\usepackage{calc}
\usepackage{amsmath}
\usepackage{amssymb}
\usepackage{amsthm}
\usepackage{amscd}
\usepackage{pictexwd,dcpic}
\usepackage{slashed}
\usepackage{mathtools}
\usepackage{braket}
\usepackage{algorithm}
\usepackage{algpseudocode} 
\usepackage[utf8]{inputenc} % allow 
\usepackage[T1]{fontenc}    % use 8-
\usepackage{hyperref} 
\usepackage{url}       
\usepackage{booktabs} 
\usepackage{amsfonts} 
\usepackage{nicefrac} 
\usepackage{microtype}
\usepackage{accents}

\newcommand{\dbtilde}[1]{\accentset{\approx}{#1}}

\newcommand{\F}{\mathcal{B}}

\newcommand{\LL}{\mathcal{L}}
\newcommand{\T}{\mathcal{T}}
\newcommand{\B}{\mathcal{B}}
\newcommand{\U}{\mathcal{U}}

\newcommand{\V}{\mathcal{V}}
\newcommand{\C}{\mathcal{C}}
\newcommand{\E}{\mathcal{E}}
\newcommand{\ones}{\mathbf{1}}

\theoremstyle{plain}
\newtheorem{theorem}{Theorem}[section]
\newtheorem{lemma}[theorem]{Lemma}

\theoremstyle{definition}
\newtheorem{definition}[theorem]{Definition}
\theoremstyle{remark}

\DeclareMathOperator*{\argmin}{arg\,min}
\title{Efficient Projection onto the Perfect Phylogeny Model}

\author{
  Bei Jia\thanks{Bei Jia is currently with Element AI.} \\
  jiabe@bc.edu\\
   \And
    Surjyendu Ray\\
  raysc@bc.edu\\
\And
\hspace{1cm}\\
\hspace{1cm}\\
  Boston College\\
\And
    Sam Safavi \\
  safavisa@bc.edu\\
\And
    Jos\'e Bento \\
 jose.bento@bc.edu\\
}

\begin{document}

\maketitle

\begin{abstract}
Several algorithms build on the \emph{perfect phylogeny model} to infer evolutionary trees.
This problem is particularly hard when evolutionary trees are inferred from the fraction of genomes that have mutations in different positions, across different samples. Existing
algorithms might do extensive searches over the space of possible trees.
At the center of these algorithms is a projection problem that assigns a fitness cost to phylogenetic trees. 
In order to perform a wide search over the space of the trees, it is critical to solve this projection problem fast.
In this paper, we use Moreau's decomposition for proximal operators, and a tree reduction scheme, to develop a new algorithm to compute this projection. Our algorithm terminates with an exact solution in a finite number of steps, and is extremely fast. In particular, it can search over all evolutionary trees with fewer than $11$ nodes, a size relevant for several biological problems (more than $2$ billion trees) in about $2$ hours.
\end{abstract}

%
%
%
%%%%%%%%%%%%%%%%%%%%%%%%%%%%
%
%
\vspace{-0.29cm}
\section{Introduction}
\vspace{-0.20cm}
The perfect phylogeny model (PPM) \cite{hudson1983properties,kimura1969number} is used in biology to study evolving populations. It assumes that the same position in the genome never mutates twice, hence mutations only accumulate.

Consider a population of organisms evolving under the PPM. The evolution process can be described by a labeled rooted tree, $T = (r,\V,\E)$, where $r$ is the root, i.e., the common oldest ancestor, the nodes $\V$ are the mutants, and the edges $\E$ are mutations acquired between older and younger mutants.
Since each position in the genome only mutates once, we can associate with each node $v \neq r$, a unique mutated position, the mutation associated to the ancestral edge of $v$. By convention, let us associate with the root $r$, a null mutation that is shared by all mutants in $T$.
This allows us to refer to each node $v\in \V$ as both a mutation in a position in the genome (the mutation associated to the ancestral edge of $v$), and a mutant (the mutant with the fewest mutations that has a mutation $v$). Hence, without loss of generality, $\V = \{1,\dots,q\}$, $\E = \{2,\dots,q\}$, where 
$q$ is the length of the genome, and $r = 1$ refers to both the oldest common ancestor and the null mutation shared by all.

One very important use of the PPM is to infer how mutants of a common ancestor evolve  \cite{el2015reconstruction,el2016multi,jiao2014inferring,malikic2015clonality, popic2015fast,satas-raphael17}. A common type of data used for this purpose is the frequency, with which different~positions in the genome mutate across multiple samples, obtained, e.g., from whole-genome or targeted deep sequencing~\cite{schuh2012monitoring}.
Consider a sample $s$, one of $p$ samples, obtained at a given stage~of the evolution process. This sample has many mutants, some with the same genome, some with different genomes. Let $F \in \mathbb{R}^{q\times p}$ be such that $F_{v,s}$ is the fraction of genomes in $s$ with a mutation in position $v$ in the genome. Let $M \in \mathbb{R}^{q\times p}$ be such that $M_{v,s}$ is the fraction of mutant $v$ in $s$. By definition, the columns of $M$ must sum to $1$. 
Let $U \in  \{0,1\}^{q\times q}$ be such that $U_{v,v'} = 1$, if and only if mutant $v$ is an ancestor of mutant $v'$, or if $v=v'$. 
We denote the set of all possible $U$ matrices, $M$ matrices and labeled rooted trees $T$, by $\mathcal{U}$, $\mathcal{M}$ and $\T$, respectively. See Figure \ref{fig:ppm_illustration} for an illustration. The PPM implies 
\vspace{-0.1cm}
\begin{equation}\label{eq:PPM_model}
F = U M.
\end{equation}
Our work contributes to the problem of inferring clonal evolution from mutation-frequencies:
\emph{ How do we infer $M$ and $U$ from  $F$?}
Note that finding $U$ is the same as finding $T$ (see Lemma \ref{th:bijection_U_T}).
\vspace{-0.1cm}
\begin{figure}[h!] \label{fig:ppm_illustration}
\begin{center}
\includegraphics[trim={0.3cm 0 0 0},clip,width=1\columnwidth]{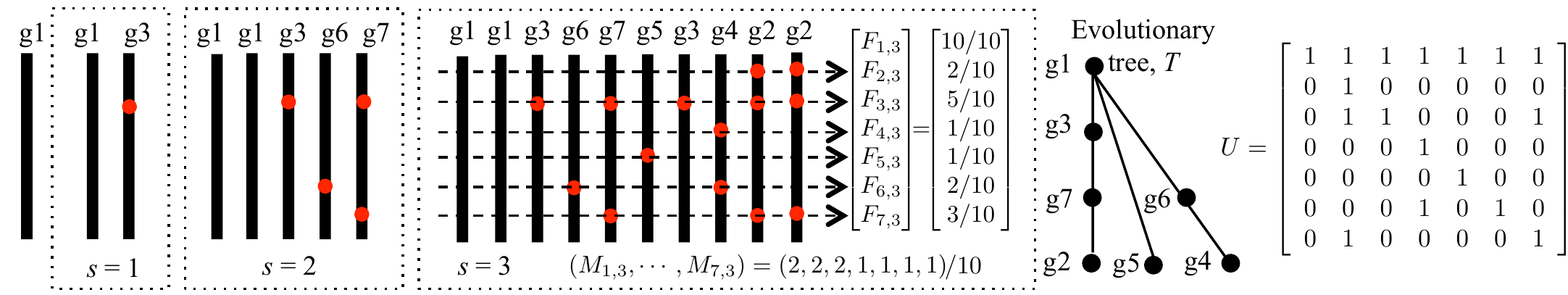}
\vspace{-0.5cm}
\caption{ \small Black lines are genomes. Red circles indicate mutations. g$i$ is the mutant with fewest mutations with position $i$ mutated. Mutation $1$, the mutation in the null position, $i = 1$, is shared by all mutants.  $g1$ is the organism before mutant evolution starts. In sample $s=3$, $2/10$ of the mutants are type $g2$, hence $M_{2,3} = 2/10$, and $3/10$ of the mutations occur in position $7$, hence $F_{7,3} = 3/10$. The tree shows the mutants' evolution.}
\vspace{-0.8cm}
\end{center}
\end{figure}

Although model \eqref{eq:PPM_model} is simple, simultaneously inferring $M$ and $U$ from 
$F$ can be hard \cite{el2015reconstruction}.
One popular inference approach is the following optimization problem over $U$, $M$ and $F$,
\vspace{-0.1cm}
\begin{align} \label{eq:projection_onto_PPM_with_tree}
& \hspace{0.3cm}\min_{U \in \mathcal{U}}  \C(U), \\
& \hspace{1.0cm}\C(U) = \min_{M,F \in \mathbb{R}^{q\times p}} \|\hat{F} - F\|   \text{ subject to } F = UM, M \geq 0,  M^\top \ones = \ones,\label{eq:projection_onto_PPM}
\vspace{-0.4cm}
\end{align}
where $\| \cdot \|$ is the Frobenius norm, and $\hat{F} \in \mathbb{R}^{q \times p}$ contains the measured fractions of mutations per position in each sample, which are known and fixed. In a nutshell, we want to project our measurement $\hat{F}$ onto the space of valid PPM models.

Problem \eqref{eq:projection_onto_PPM_with_tree} is a hard
mixed integer-continuous optimization problem.
To approximately solve  it, we might  
find a finite subset $\{U_i\} \subset \mathcal{U}$, that corresponds to a ``heuristically good'' subset of trees, $\{T_i\}\subset \T$, and, for each fixed  matrix $U_i$, solve  \eqref{eq:projection_onto_PPM}, which is a convex optimization problem. We can then return $T_x$, where $x  \in {\argmin_i \C(U_i)}$.
Fortunately, in many biological applications, e.g., \cite{el2015reconstruction,el2016multi,jiao2014inferring,malikic2015clonality,popic2015fast,satas-raphael17}, the reconstructed evolutionary tree involves a very small number of mutated positions, e.g., $q \leq 11$. In practice, a position $v$ might be an \emph{effective} position that is a cluster of multiple real positions in the genome. 
%Although the number of positions in any genome is much larger than $13$, the number of positions relevant for a given experiment is often much smaller and, further more, these positions are often clustered into even fewer effective positions, on which tree inference is finally performed.
%
For a small $q$, we can compute $\C(U)$ for 
many trees, and hence approximate $M$, $U$, and get uncertainty measures for these estimates. This is important, since data is generally scarce and noisy.  
%\footnote{Recall that Cayley's formula  implies that the number of labeled trees on $q$ nodes grows as $q^{q-2}$ \cite{cayley1889theorem}. We also recall that we can assume we know the root of $\T$, which is the organism before the evolution experiment begins.}.

{\bf{Contributions:}}
(i) we propose a new algorithm to compute $\C(U)$ exactly in $\mathcal{O}(q^2 p)$ steps, the first non-iterative algorithm to compute $\C(U)$;
(ii) we compare its performance against state-of-the-art iterative algorithms, and observe a much faster convergence. In particular, our algorithm scales much faster than  $\mathcal{O}(q^2 p)$ in practice; (iii) we implement our algorithm on a GPU, and show that it computes the cost of all (more than $2$ billion) trees with $\leq 11$ nodes, in $\leq 2.5$ hours.

%
%
%
%%%%%%%%%%%%%%%%%%%%%%%%%%%%%%%%%%
%%%%%%%%%%%%%%%%%%%%%%%%%%%%%%%%%%%
%%%%%%%%%%%%%%%%%%%%%%%%%%%%%%%%%%%%
%
%
\vspace{-0.3cm}
\section{Related work}
\vspace{-0.3cm}

A problem related to ours, but somewhat different, is that of inferring a phylogenetic
tree from single-cell whole-genome sequencing data.
Given all the mutations in a set of mutants, the problem is to arrange the mutants in a phylogenetic tree, \cite{fernandez2001perfect,gusfield1991efficient}.
% Here, the input are all the mutations in a set of mutants, and we want to arrange the mutants in a phylogenetic tree \cite{fernandez2001perfect,gusfield1991efficient}.
Mathematically, this corresponds to
inferring $T$ from partial or corrupted observation of $U$. If the PPM is assumed, and all the mutations of all the mutants are correctly observed, this problem can be solved in linear time, e.g., \cite{ding2006linear}. In general, this problem is equivalent to finding a minimum cost Steiner tree on a hypercube, whose nodes and edges represent mutants and mutations respectively, a problem known to be hard~\cite{garey2002computers}. 

We mention a few works on clonality inference, based on the PPM,  that try to infer both $U$ and $M$ from $\hat{F}$. No previous work solves problem \eqref{eq:projection_onto_PPM_with_tree} exactly in general, even for trees of size $q \leq 11$. Using our fast projection algorithm, we can solve \eqref{eq:projection_onto_PPM_with_tree} exactly by searching over all trees, if $q \leq 11$. Ref.~\cite{el2015reconstruction} (\emph{AncesTree}) 
 reduces the space of possible trees $\T$ to subtrees of a heuristically constructed DAG. 
The authors use the element-wise $1$-norm in \eqref{eq:projection_onto_PPM} and, after introducing more variables to linearize the product $UM$, reduce this search to solving a MILP, which they try to solve via branch and bound. Ref. \cite{malikic2015clonality} (\emph{CITUP}) searches the space of all \emph{unlabeled} trees, and, for each unlabeled tree, tries to solve an MIQP, again using branch and bound techniques, which finds a labeling for the unlabeled tree, and simultaneously minimizes the distance $\|\hat{F}-F\|$. Refs. \cite{jiao2014inferring} and \cite{deshwar2015phylowgs} (\emph{PhyloSub/PhyloWGS}), use a stochastic model to sample trees that are likely to explain the data. Their  model is based on \cite{ghahramani2010tree}, which  generates hierarchical clusterings of objects, and from which lineage trees can be formed. A score is then computed for these trees, and the highest scoring trees are returned.

Procedure \eqref{eq:projection_onto_PPM_with_tree} can be justified as MLE if we assume the stochastic model $\hat{F} = F + \mathcal{N}(0,I\sigma^2)$, where $F$, $U$ and $M$ satisfy the PPM model, and 
 $\mathcal{N}(0,I\sigma^2)$ represents additive, component-wise, Gaussian measurement noise, with zero mean and covariance $I\sigma^2$.
Alternative stochastic models can be assumed, e.g., as $M - U^{-1} \hat{F} = \mathcal{N}(0,I\sigma^2)$, where $M$ is non-negative and its columns must sum to one, and  
$\mathcal{N}(0,I\sigma^2)$ is as described before.
For this model, and for each matrix $U$, the cost $\C(U)$ is a projection of 
$U^{-1} \hat{F}$ onto the probability simplex $M \geq 0, M^\top \ones = \ones$. Several fast algorithms are known for this problem, e.g., 
\cite{condat2016fast,duchi2008efficient,gong2011efficient,liu2009efficient,michelot1986finite} and references therein.
In a $pq$-dimensional space, the exact projection onto the simplex can be done in $\mathcal{O}(qp)$ steps.
%
%The worst case complexity of our algorithm is, for $p=1$, $\mathcal{O}(q^2)$, and we can prove an average complexity of $\mathcal{O}(q \log q)$. We conjecture that the best possible, worst-case, complexity is $\mathcal{O}(q)$. Our complexity is linearly with the number of samples $p$.

Our algorithm is the first to solve \eqref{eq:projection_onto_PPM} exactly in a finite number of steps. We can also use iterative methods to solve \eqref{eq:projection_onto_PPM}. One advantage
of our algorithm is that it has no tuning parameters, and requires no effort to check for convergence for a given accuracy.
Since iterative algorithms can converge very fast,  we numerically compare the speed of our algorithm with different implementations of the Alternating Direction Method of Multipliers (ADMM) \cite{boyd2011distributed}, which, if properly tuned, has a convergence rate that equals the fastest convergence rate among all first order methods~\cite{francca2016explicit}
%has the fastest possible convergence rate among all first-order methods
under some convexity assumptions, and is known to produce good solutions for several other kinds of problems, even for non-convex ones \cite{hao2016testing,francca2017distributed,laurenceinvFBA2018,mathysparta,zoran2014shape,bento2015proximal,bento2013message}.

%
%
%
%%%%%%%%%%%%%%%%%%%%%%%%%%%%%%%%%%
%%%%%%%%%%%%%%%%%%%%%%%%%%%%%%%%%%%
%%%%%%%%%%%%%%%%%%%%%%%%%%%%%%%%%%%%
%
%
\vspace{-0.2cm}
\section{Main results}\label{sec:main_results}
\vspace{-0.2cm}

We now state our main results, and explain the ideas  behind their proofs. Detailed proofs can be found in the Appendix.

Our algorithm computes $\C(U)$ and minimizers of \eqref{eq:projection_onto_PPM}, resp. $M^*$ and $F^*$, by solving an equivalent problem.
Without loss of generality, we assume that $p = 1$, since, by squaring the objective in \eqref{eq:projection_onto_PPM}, it decomposes into $p$ independent problems.
Sometimes we denote $\C(U)$ by 
$\C(T)$, since given $U$, we can specify $T$, and vice-versa. Let $\bar{i}$ be the closest ancestor of $i$ in $T= (r,\V,\E)$. 
Let $\Delta i$ be the set of all the ancestors of $i$ in $T$, plus $i$. Let $\partial i$ be the set of children of $i$ in $T$.
\begin{theorem}[Equivalent formulation] \label{th:dual_prob}
Problem \eqref{eq:projection_onto_PPM} can be solved by solving
%is equivalent to
%
\vspace{-0.1cm}
\begin{align} 
& \min_{t\in \mathbb{R}} \;\; t + \mathcal{L}(t),\label{eq:dual_higher}\\[-0.4cm]
& \hspace{1.46cm}\mathcal{L}(t) = \min_{Z\in\mathbb{R}^q} \frac{1}{2}\sum_{i \in \V} (Z_i - Z_{\bar{i}})^2 \text{ subject to } Z_i \leq t - N_i \;,\forall i\in \V \label{eq:dual},
\end{align}
where $N_i= \sum_{j \in \Delta i} \hat{F}_j$, and, by convention, $Z_{\bar{i}} = 0$ for $i = r$.
In particular, if $t^*$ minimizes 
\eqref{eq:dual_higher}, 
$Z^*$ minimizes \eqref{eq:dual} for $t=t^*$, and $M^*,F^*$ minimize \eqref{eq:projection_onto_PPM}, then %
\begin{equation}\label{eq:Z_star_M_star_relation}
M^*_i = -Z^*_i + Z^*_{\bar{i}} + \sum_{r \in \partial i} (Z^*_r - Z^*_{\bar{r}}) \text{ and } F^*_i = -Z^*_i + Z^*_{\bar{i}}, \forall i\in \V.
\end{equation}
Furthermore, $t^*$, $M^*$, $F^*$ and $Z^*$ are unique.
\end{theorem}
\vspace{-0.2cm}
Theorem \ref{th:dual_prob} comes from a dual form of \eqref{eq:projection_onto_PPM}, which we build using Moreau's decomposition \cite{moreau1962decomposition}.

%
%%%%%%%%%%%%%%%%%%%%%%%%%%%%%%%%%%%%%%%%%%%%%%%%%%
%
\vspace{-0.3cm}
\subsection{Useful observations}\label{sec:useful_obs}

Let $Z^*(t)$ be the unique minimizer of \eqref{eq:dual} for some $t$. 
The main ideas behind our algorithm depend on a few simple properties of the paths $\{Z^*(t)\}$ and $\{\LL'(t)\}$, the derivative of $\LL(t)$ with respect to $t$. Note that $\mathcal{L}$ is also a function of $N$, as defined in Theorem \ref{th:dual_prob}, which depends on the input data $\hat{F}$.
\begin{lemma}\label{th:convexity_of_L_dual}
$\LL(t)$ is a convex function of $t$ and $N$. Furthermore, $\LL(t)$ is continuous in $t$ and $N$, and $\LL'(t)$ is non-decreasing with $t$. 
\end{lemma}
\begin{lemma}\label{th:continuity_of_Z_t}
$Z^*(t)$ is continuous as a function of $t$ and $N$. $Z^*(t^*)$ is continuous as a function of $N$.
\end{lemma}
Let
$
\B(t) = \{i: Z^*(t)_i = t - N_i\},
$ %change Z^*(t)_i to Z_i^*(t)????
i.e., the set of components of the solution at the boundary of \eqref{eq:dual}. Variables in $\B$ are called \emph{fixed}, and we call other variables \emph{free}. Free (resp. fixed) nodes are nodes corresponding to free (resp. fixed) variables. 
\begin{lemma}\label{th:B_piece_wise_constant}
$\B(t)$ is piecewise constant in $t$.
\end{lemma}
Consider dividing the tree $T=(r,\V,\E)$ into subtrees, each with at least one free node,
 using $\B(t)$ as separation points. See Figure \ref{fig:tree_and_variables} in Appendix
 \ref{sec:app:further_illustrations} for an illustration. Each $i\in \B(t)$ 
 belongs to at most $\text{degree}(i)$ different subtrees, where $\text{degree}(i)$ is the degree of node $i$, and each
 $i \in \V \backslash \B(t)$ belongs exactly to one subtree.
  Let $T_1,\dots, T_k$ be the set of resulting (rooted, labeled) trees. Let $T_w = (r_w,\V_w,\E_w)$, where the root $r_w$ is the  closest node in $T_w$ to $r$.
We call $\{T_w\}$ the subtrees \emph{induced by} $\B(t)$. We define $\B_w(t) = \B(t) \cap \V_w$, and, when it does not create ambiguity, we drop the index $t$ in $\B_w(t)$. Note that different $\B_w(t)$'s might have elements in common.  Also note that, by construction, if $i \in \B_w$, then $i$ must be a leaf of $T_w$, or the root of $T_w$.

\begin{definition}\label{def:sub_problem_Tw_Bw}
The $(T_w,\F_w)$-problem is the optimization problem over $|\V_w \backslash \B(t)|$ variables
\begin{align}\label{eq:simpler_sub_problem}
&\min_{\{Z_j : j \in \V_w \backslash \B(t)\}} \;(1/2)\sum_{j \in \V_w } (Z_j - Z_{\bar{j}})^2,
\end{align}
where $\bar{j}$ is the parent of $j$ in $T_w$, $Z_{\bar{j}} = 0$ if ${j} = r_w$, and  $Z_j =Z^*(t)_j= t - N_j$ if $j \in \B_w(t)$.
\end{definition}
\begin{lemma}\label{eq:decomposition_of_problem}
Problem \eqref{eq:dual} decomposes into $k$ independent problems. In particular, the minimizers $\{Z^*(t)_j : j \in \V_w \backslash \B(t) \}$ are determined 
as the solution~of the $(T_w,\F_w)$-problem.
If $j \in \V_w$, then $Z^*(t)_j = c_1 t + c_2$ , where $c_1$ and $c_2$ depend on $j$ but not on $t$, and $0 \leq c_1 \leq 1$.
\end{lemma}
\begin{lemma}\label{th:Z_and_Lprime_are_piece_wise_linear}
$Z^*(t)$ and $\LL'(t)$ are piecewise linear and continuous in $t$. Furthermore, $Z^*(t)$ and $\LL'(t)$ change linear segments if and only if $\B(t)$ changes.
\end{lemma}
\begin{lemma}\label{eq:B_always_grows}
If $t \leq t'$, then $\B(t') \subseteq \B(t)$. In particular, $\B(t)$ changes at most $q$ times with $t$.
\end{lemma}
\begin{lemma}\label{th:Z_and_Lprime_finite_break_points}
$Z^*(t)$ and $\LL'(t)$ have less than $q+1$ different linear segments.
\end{lemma}

%
%
%
%%%%%%%%%%%%%%%%%%%%%%%%%%%%%%%%%%
%%%%%%%%%%%%%%%%%%%%%%%%%%%%%%%%%%%
%%%%%%%%%%%%%%%%%%%%%%%%%%%%%%%%%%%%
%
%

\vspace{-0.4cm}
\subsection{The Algorithm}\label{sec:main_alg}

In a nutshell, our algorithm computes the solution path $\{Z^*(t)\}_{t \in \mathbb{R}}$ and the  derivative $\{\LL'(t)\}_{t \in \mathbb{R}}$. From these paths, it finds the unique  $t^*$, at which 
\begin{equation}\label{eq:condition_for_t}
{\rm d}(t + \mathcal{L}(t))/{{\rm d}t} =0\rvert_{t = t^*} \Leftrightarrow
\mathcal{L}'(t^*) = -1.
\end{equation}
It then evaluates the path $Z^*(t)$ at $t =t^*$, and uses this value, along with \eqref{eq:Z_star_M_star_relation}, to find $M^*$ and $F^*$, the unique minimizers of \eqref{eq:projection_onto_PPM}. Finally, we compute $\C(T) = \|\hat{F} - F^*\|$.

We know that $\{Z^*(t)\}$ and $\{\LL'(t)\}$ are continuous piecewise linear, with a finite number of different linear segments (Lemmas \ref{th:Z_and_Lprime_are_piece_wise_linear}, \ref{eq:B_always_grows} and \ref{th:Z_and_Lprime_finite_break_points}).
Hence, to describe   
$\{Z^*(t)\}$ and $\{\LL'(t)\}$, we only need to evaluate them at the \emph{critical values}, $t_1>t_2 >  \dots >t_k$, at which $Z^*(t)$ and $\LL'(t)$ change linear segments. 
We will later use Lemma \ref{th:Z_and_Lprime_are_piece_wise_linear} 
as a criteria to find the critical values. Namely, $\{t_i\}$ are the values of $t$ at which, as $t$ decreases, new variables become fixed, and $\mathcal{B}(t)$ changes. Note that variables never become free once fixed, by Lemma \ref{eq:B_always_grows}, which also implies that $k \leq q$.

The values $\{Z^*(t_i)\}$ and $\{\LL'(t_i)\}$ are computed sequentially as follows. If $t$ is very large, the constraint in \eqref{eq:dual} is not active, and $Z^*(t) = \LL(t) = \LL'(t)=0$. Lemma \ref{th:Z_and_Lprime_are_piece_wise_linear} tells us that, as we decrease $t$, the first critical value is the largest $t$ for which this constraint becomes active, and at which $\mathcal{B}(t)$ changes for the first time. Hence, if
$i= 1$, we have $t_i = \max_s\{N_s\}$, $Z^*(t_i) = \LL'(t_i) = 0$, and $\B(t_i) = \arg\max_s\{N_s\}$.
Once we have $t_i$,  we compute the rates $Z'^*(t_i)$ and $\LL''(t_i)$ from $\B(t_i)$ and $T$, as explained in Section \ref{sec:computing_rates}. Since the paths are piecewise linear, derivatives are not defined at critical points. Hence, here, and throughout this section, these derivatives are taken from the left, i.e., $Z'^*(t_i) = \lim_{t \uparrow t_i} (Z^*(t_i) - Z^*(t))/(t_i - t)$ and $\LL''(t_i) = \lim_{t \uparrow t_i} (\LL'(t_i) - \LL'(t))/(t_i - t)$.

Since $Z'^*(t)$ and $\LL''(t)$ are constant for $t\in (t_{i+1}, t_{i}]$, for 
$t\in (t_{i+1}, t_i]$ we have
\begin{align} \label{eq:liner_relation_between_critical_points}
Z^*(t) = Z^*(t_i) + (t - t_i)Z'^*(t_i), \quad \LL'(t) = \LL'(t_i) + (t - t_i)\LL''(t_i),
\end{align}
and the next critical value, $t_{i+1}$, is the largest $t < t_i$, for which new variables become fixed, and $\B(t)$ changes. The value $t_{i+1}$ is found by solving for $t<t_i$ in 
\begin{equation}\label{eq:formula_for_next_t}
Z^*(t)_r = Z^*(t_i)_r + (t - t_i)Z'^*(t_i)_r = t - N_r,
\end{equation}
and keeping the largest solution among  all $r \notin \B$.
Once $t_{i+1}$ is computed, we update $\B$ with the new variables that became fixed, and we obtain $Z^*(t_{i+1})$ and $\LL'(t_{i+1})$ from \eqref{eq:liner_relation_between_critical_points}.
The process then repeats.

By Lemma \ref{th:convexity_of_L_dual}, $\LL'$ never increases. Hence, we stop this process (a) 
as soon as $\LL'(t_i) < -1$, or (b) when all the variables are in $\B$, and thus there are no more critical values to compute. If (a), let $t_k$ be the last critical value with $\LL'(t_k) > -1$, and if (b), let $t_k$ be the last computed critical value. We use $t_k$ and \eqref{eq:liner_relation_between_critical_points} to compute $t^*$, at which $\LL'(t^*) = -1$ and also $Z^*(t^*)$. From $Z^*(t^*)$ we then compute $M^*$ and $F^*$ and $\C(U) = \|\hat{F} - F^*\|$.

The algorithm is shown compactly in Alg. \ref{alg:projection}. Its inputs are $\hat{F}$ and $T$, represented, e.g., using a linked-nodes data structure. Its outputs are minimizers to \eqref{eq:projection_onto_PPM}. It makes use of a procedure \emph{ComputeRates}, which we will explain later. This procedure terminates in $\mathcal{O}(q)$ steps and uses $\mathcal{O}(q)$ memory. Line \ref{alg:line:next_critical_value} comes from solving \eqref{eq:formula_for_next_t} for $t$. In line \ref{alg:line:return}, the symbols $M^*(Z^*, T)$ and $F^*(Z^*, T)$ remind us that $M^*$ and $F^*$ are computed from $Z^*$ and $T$ using \eqref{eq:Z_star_M_star_relation}.
The correctness of Alg. \ref{alg:projection} follows from the Lemmas in Section \ref{sec:useful_obs}, and the explanation above. In particular, since there are at most $q+1$ different linear regimes, the bound $q$ in the for-loop does not prevent us from finding any critical value. Its time complexity is $\mathcal{O}(q^2)$, since each line completes in $\mathcal{O}(q)$ steps, and is executed at most $q$ times. 
\begin{theorem}[Complexity]\label{th:main_alg_complexity}
Algorithm \ref{alg:projection} finishes in $\mathcal{O}(q^2)$ steps, and requires $\mathcal{O}(q)$ memory.
\end{theorem}
\begin{theorem}[Correctness]\label{th:alg_correctness}
Algorithm \ref{alg:projection} outputs the solution to \eqref{eq:projection_onto_PPM}.
\end{theorem}
\vspace{-0.4cm}
\begin{algorithm}[h!]
\caption{Projection onto the PPM (input: $T$ and $\hat{F}$; output: $M^*$ and $F^*$)} \label{alg:projection}
\begin{algorithmic}[1]
	\State $N_i = \sum_{j \in \Delta i} \hat{F}_j$, for all $i\in \V$ 
    \label{alg:line:compute_N_from_F}
    \Comment{This takes $\mathcal{O}(q)$ steps using a DFS, see proof of Theorem \ref{th:main_alg_complexity}}
	\State $i = 1$, $t_i = \max_r\{N_r\}$, $\B(t_i) = \arg \max_r\{N_r\}$, $Z^*(t_i) = {\bf 0}$,  $\LL'(t_i) = {0}$.\label{alg:line:init}
    \Comment{Initialize}
	\For{$i = 1$ to $q$} \label{alg:line:forloop}
		\State $(Z'^*(t_i),\LL''(t_i)) = \text{ComputeRates}(\B(t_i),T)$
                \Comment{Update rates of change}\label{alg:line:main_alg_compute_rates}
        \State $P = \{P_r : P_r = \frac{N_r + Z^*(t_i)_r - t_i Z'^*(t_i)_r}{1 - Z'^*(t_i)_r} \text{ if } r \notin \B(t_i), t_r < t_i, \text{ and } P_r = -\infty \text{ otherwise}\}$\label{alg:line:next_critical_value}
        \State $t_{i+1} = \max_r P_r$
        \Comment{Update next critical value from \eqref{eq:liner_relation_between_critical_points}}
		\State $\B(t_{i+1}) = \B(t_{i}) \cup \arg \max_r P_s$ \label{alg:line:main_alg_new_fixed_nodes}
        \Comment{Update list of fixed variables}
        \State $Z^*(t_{i+1}) = Z^*(t_{i}) + (t_{i+1} - t_i) Z'^*(t_i)$
        \Comment{Update solution path }
        \State $\LL'(t_{i+1}) = \LL'(t_{i}) + (t_{i+1} - t_i) \LL''(t_i)$ 
        \Comment{Update objective's derivative}
                \State {\bf if }$\LL'(t_{i+1}) < -1 $ {\bf then break}
               \Comment{If already passed by $t^*$, then exit the for-loop}
	\EndFor
	\State $t^* = t_{i} - \frac{1 + \LL'(t_i)}{\LL''(t_i)}$\label{alg:line:find_tstar}
    \Comment{Find solution to \eqref{eq:condition_for_t}}
    \State $Z^* = Z^*(t_i) + (t^* - t_i) Z'^*(t_i)$\label{alg:line:find_Zstar}
    \Comment{Find minimizers of \eqref{eq:dual} for $t = t^*$}
    \State \textbf{return} $M^*(Z^*,T)$, $F^*(Z^*,T)$ \label{alg:line:return}
    \Comment{Return solution to \eqref{eq:projection_onto_PPM} using \eqref{eq:Z_star_M_star_relation}, which takes $\mathcal{O}(q)$ steps}
\end{algorithmic}
\end{algorithm}
%
%
%
%%%%%%%%%%%%%%%%%%%%%%%%%%%%%%%%
%
%

\vspace{-0.6cm}
\subsection{Computing the rates}\label{sec:computing_rates}
\vspace{-0.2cm}

We now explain how the procedure \emph{ComputeRates} works. Recall that it takes as input the tree $T$ and the set $\B(t_i)$, and it outputs the derivatives $Z'^*(t_i)$ and $\LL''(t_i)$. 

A simple calculation shows that if we compute $Z'^*(t_i)$, then computing $\LL''(t_i)$ is easy.
\begin{lemma} \label{th:computing_ddL_from_dZ}
$\LL''(t_i)$ can be computed from $Z'^*(t_i)$ in $\mathcal{O}(q)$ steps and with $\mathcal{O}(1)$ memory as
\begin{equation} \label{eq:computing_ddL_from_dZ}
\LL''(t_i) = \sum_{j \in \V} (Z'^*(t_i)_j - Z'^*(t_i)_{\bar{j}})^2,
\end{equation}
\end{lemma}
where $\bar{j}$ is the closest ancestor to $j$ in $T$.
We note that if $j \in \B(t_i)$, then, by definition, $Z'^*(t_i)_j= 1$. Assume now that  $j \in \V \backslash \B(t_i)$.
Lemma \ref{eq:decomposition_of_problem} implies we can find $Z'^*(t_i)_j$ by solving the $(T_w = (r_w,\V_w,\E_w),\B_w)$-problem as a function of $t$, where $w$ is such that $j \in \V_w$. 
In a nutshell, \emph{ComputeRates} is a recursive procedure to solve all the $(T_w,\B_w)$-problems as an explicit function~of~$t$.

It suffices to explain how \emph{ComputeRates} solves one particular $(T_w,\B_w)$-problem explicitly. To simplify notation, in the rest of this section, we refer to $T_w$ and $\B_w$ as $T$ and $\B$. Recall that, by the definition of $T=T_w$ and $\B=\B_w$, if $i \in \B$, then $i$ must be a leaf of $T$, or the root of $T$.
\begin{definition} \label{def:sub_problem_T_B_A_B_G}
Consider a rooted tree $T = (r,\V,\E)$, a set $\B\subseteq\V$, and variables
$\{Z_j:j\in \V\}$ such that, if 
$j\in \B$, then $Z_j = \alpha_j t + \beta_j$ for some $\alpha$ and $\beta$.
We define the $(T,\B,\alpha,\beta,\gamma)$-problem as
\begin{align}\label{eq:general_sub_problem}
\min_{\{Z_j: j \in \V \backslash \B\}} \frac{1}{2} \sum_{j \in \V} \gamma_{j} (Z_j - Z_{\bar{j}})^2,
\end{align}
where $\gamma > 0$, $\bar{j}$ is the closest ancestor to $j$ in $T$,
and $Z_{\bar{j}} = 0$ if ${j}=r$.
\end{definition}
We refer to the solution of the $(T,\B,\alpha,\beta,\gamma)$-problem as $\{Z^*_j: j \in \V \backslash \B\}$, which uniquely minimizes \eqref{eq:general_sub_problem}.
Note that \eqref{eq:general_sub_problem} is unconstrained and its solution, $Z^*$, is a linear function of $t$. Furthermore, the $(T_w,\B_w)$-problem is the same as the $(T_w,\B_w,{\bf 1},-N,{\bf 1})$-problem, which is what we actually solve.

We now state three useful lemmas that help us solve any $(T,\B,\alpha,\beta,\gamma)$-problem efficiently. 

\begin{lemma}[Pruning]\label{th:pruning}
Consider the solution $Z^*$ of the $(T,\B,\alpha,\beta,\gamma)$-problem.
Let $j\in\V \backslash \B$ be a leaf.
Then $Z^*_j = Z^*_{\bar{j}}$. Furthermore, consider the  $(\tilde{T},\B,\alpha,\beta,\gamma)$-problem, where $\tilde{T} = (\tilde{r},\tilde{\V},\tilde{\E})$ is equal to $T$ with node $j$ pruned, and let its solution be $\tilde{Z}^*$. We have that $Z^*_i = \tilde{Z}^*_i$, for all $i \in \tilde{\V}$.
\end{lemma}
\begin{lemma}[Star problem]\label{th:star_problem}
Let $T$ be a star
such that node $1$ is the center node, node $2$ is the root, and 
nodes $3,\dots, r$ are leaves.
Let $\B = \{2,\dots, r\}$. Let $Z^*_1 \in \mathbb{R}$ be the solution
of the $(T,\B,\alpha,\beta,\gamma)$-problem.
Then,
\begin{equation} \label{eq:star_solution}
Z^*_1 = \left( \frac{\gamma_1 \alpha_2+ \sum^r_{i=3} \gamma_r \alpha_r}{\gamma_1 + \sum^r_{i=3} \gamma_r}\right) t + \left( \frac{\gamma_1 \beta_2+ \sum^r_{i=3} \gamma_r \beta_r}{\gamma_1 + \sum^r_{i=3} \gamma_r}\right).
\end{equation}
In particular, to find the rate at which $Z^*_1$ changes with $t$, we only need to know $\alpha$ and $\gamma$,  not  $\beta$.
\end{lemma}
\begin{lemma}[Reduction]\label{th:tree_reduction}
Consider the $(T,\B,\alpha,\beta,\gamma)$-problem
such that $j,\bar{j} \in \V \backslash \B$, and such that
$j$ has all its children $1,\dots,r \in \B$. Let $Z^*$ be its solution.
Consider the $(\tilde{T},\tilde{\B},\tilde{\alpha},\tilde{\beta},\tilde{\gamma})-problem$,
where $\tilde{T} = (\tilde{r},\tilde{\V},\tilde{\E})$ is equal to $T$ with nodes $1,\dots,r$ removed, and
$\tilde{\B} = (\B \backslash \{1,\dots,r\}) \cup \{j\}$. Let ${\tilde{Z}^*}$ be its solution.
If $(\tilde{\alpha}_i,\tilde{\beta}_i,\tilde{\gamma}_i)=(\alpha_i,\beta_i,\gamma_i)$ for all $i \in \B \backslash \{1,\ldots,r\}$, and $\tilde{\alpha}_j$, $\tilde{\beta}_j$ and $\tilde{\gamma}_j$ satisfy
\begin{align}
&\tilde{\alpha}_j = \frac{\sum^r_{i=1}  \gamma_i \alpha_i}{\sum^r_{i=1} \gamma_i}, \;\; \tilde{\beta}_j = \frac{\sum^r_{i=1}  \gamma_i \beta_i} {\sum^r_{i=1}\gamma_i}, \tilde{\gamma}_j = \left((\gamma_j)^{-1} + \left(\sum^r_{i=1} \gamma_i\right)^{-1}\right)^{-1}\label{eq:reduction_equations},
\end{align}
then $Z^*_i = {\tilde{Z}^*}_i$ \,for all $i \in \V \backslash \{j\}$.
\end{lemma}

Lemma \ref{th:star_problem} and Lemma \ref{th:tree_reduction} allow us to recursively solve  
any $(T,\B,\alpha,\beta,\gamma)$-problem, and obtain for it an explicit solution of the form $Z^*(t) = c_1 t + c_2$, where $c_1$ and $c_2$ do not depend on $t$.

Assume that we have already repeatedly pruned $T$, by repeatedly invoking Lemma \ref{th:pruning}, such that, if $i$ is a leaf, then $i \in \B$. See Figure \ref{fig:recursion}-(left).
First, we find some node $j \in \V \backslash \B$ such that all of its children are in $\B$.
If $\bar{j} \in \B$, then $\bar{j}$ must be the root, and the $(T,\B,\alpha,\beta,\gamma)$-problem
must be a star problem as in Lemma \ref{th:star_problem}. We can use Lemma \ref{th:star_problem} to solve it explicitly. 
Alternatively, if $\bar{j} \notin \V \backslash \B$, then we  invoke Lemma \ref{th:tree_reduction}, and reduce the $(T,\B,\alpha,\beta,\gamma)$-problem to a strictly smaller  $(\tilde{T},\tilde{\B},\tilde{\alpha},\tilde{\beta},\tilde{\gamma})$-problem, which we solve  recursively. Once the $(\tilde{T},\tilde{\B},\tilde{\alpha},\tilde{\beta},\tilde{\gamma})$-problem is solved, we have an explicit expression $Z^*_i(t) = {c_1}_i t + {c_2}_i$ for all $i \in \V \backslash \{j\}$, and, in particular, we have an explicit expression   $Z^*_{\bar{j}}(t) = {c_1}_{\bar{j}} t + {c_2}_{\bar{j}}$. The only  free variable of the $(T,\B,\alpha,\beta,\gamma)$-problem to be determined is $Z^*_j(t)$. To compute $Z^*_j(t)$, we apply Lemma \ref{th:star_problem} to the $({\dbtilde{T}},{\dbtilde{\B}},{\dbtilde{\alpha}},{\dbtilde{\beta}},{\dbtilde{\gamma}})$-problem, where ${\dbtilde{T}}$ is a star around $j$, ${\dbtilde{\gamma}}$ are the components of $\gamma$ corresponding to nodes that are neighbors of  $j$, ${\dbtilde{\alpha}}$ and ${\dbtilde{\beta}}$ are such that $Z^*_i(t)  = {\dbtilde{\alpha}}_i t + {\dbtilde{\beta}}_i$ for all $i$ that are neighbors of $j$, and for which $Z^*_i(t)$ is already known, and ${\dbtilde{\B}}$ are all the neighbors of $j$. See Figure \ref{fig:recursion}-(right).

The algorithm is compactly described in Alg. \ref{alg:compute_rates_rec}. It is slightly different from the description above for computational efficiency. Instead of computing $Z^*(t) = c_1t+c_2$, we keep track only of $c_1$, the rates, and we do so only for the variables in $\V \backslash \B$.
The algorithm assumes that the input $T$ has been pruned. The inputs $T$, $\B$, $\alpha$, $\beta$ and $\gamma$ are passed by reference. They are modified inside the algorithm but, once \emph{ComputeRatesRec} finishes, they keep their initial values. Throughout the execution of the algorithm, $T = (r,\V,\E)$ encodes (1) a doubly-linked list where each node points to its children and its parent, which we call $T.a$, and (b) a a doubly-linked list of all the nodes in $\V \backslash \B$ for which all the children are in $\B$, which we call $T.b$. In the proof of Theorem \ref{th:complexity_of_rec_reduce}, we prove how this  representation of $T$ can be kept updated with  little computational effort. 
The input $Y$, also passed by reference, starts as an uninitialized array of size $q$, where we will store the rates $\{Z'^*_i\}$. At the end, we read $Z'^*$ from $Y$.
\begin{algorithm}[h!]
\caption{ComputeRatesRec (input: $T=(r,\V,\E), \B,\alpha,\beta,\gamma,Y$)} \label{alg:compute_rates_rec}
\begin{algorithmic}[1]
    \State Let $j$ be some node in $\V \backslash \B$ whose children are in $\B$\label{alg:line:rec:1}
    \Comment{We read $j$ from $T.b$ in $\mathcal{O}(1)$ steps}
    \If{$\bar{j} \in \B$}\label{alg:line:rec:2}
    \State  Set $Y_j$ using \eqref{eq:star_solution} in Lemma \ref{th:star_problem} \label{alg:line:rec:3}
    \Comment{If $\bar{j} \in \B$, then the $(T,\B,\alpha,\beta,\gamma)$-problem is star-shaped}
    \Else
    \State Modify $(T,\B,\alpha,\beta,\gamma)$ to match  $(\tilde{T},\tilde{\B},\tilde{\alpha},\tilde{\beta},\tilde{\gamma})$ defined by Lemma \ref{th:tree_reduction} 
    for $j$ in line \ref{alg:line:rec:1} \label{alg:line:rec:5}
    \State  $\text{ComputeRatesRec}(T,\B,\alpha,\beta,\gamma,Y)$\label{alg:line:rec:6}
    \Comment{Sets $Y_i= Z'^*_i$ for all ${i \in \V \backslash \B}$; $Y_j$ is not yet defined}
    \State Restore $(T,\B,\alpha,\beta,\gamma)$ to its original value before line \ref{alg:line:rec:5} was executed \label{alg:line:rec:7}
\State Compute $Y_j$ from \eqref{eq:star_solution}, using for $\alpha,\beta,\gamma$ in \eqref{eq:star_solution} the values
    ${\dbtilde{\alpha}},{\dbtilde{\beta}},{\dbtilde{\gamma}}$, where ${\dbtilde{\gamma}}$ are the components of $\gamma$ corresponding to nodes that are neighbors of  $j$ in $T$, and ${\dbtilde{\alpha}}$ and ${\dbtilde{\beta}}$ are such that $Z^*_i  = {\dbtilde{\alpha}}_i t + {\dbtilde{\beta}}_i$ for all $i$ that are neighbors of $j$ in $T$, and for which $Z^*_i$ is already known\label{alg:line:rec:8}
    \EndIf
\end{algorithmic}
\end{algorithm}

Let $q$ be the number of nodes of the tree $T$ that is the input at the zeroth level of the recursion. 
\begin{theorem}\label{th:complexity_of_rec_reduce}
Algorithm \ref{alg:compute_rates_rec} correctly computes $Z'^*$ for the $(T,\B,\alpha,\beta,\gamma)$-problem, and it
can be implemented to finish in $\mathcal{O}(q)$ steps, and to use $\mathcal{O}(q)$ memory.
\end{theorem}
The correctness of Algorithm \ref{alg:compute_rates_rec} follows from Lemmas \ref{th:pruning}-\ref{th:tree_reduction}, and the explanation above. Its complexity is bounded by the total time spent on the two lines that actually compute rates during the whole recursion, lines \ref{alg:line:rec:3} and \ref{alg:line:rec:8}. All the other lines only transform the input problem into a more computable form. Lines \ref{alg:line:rec:3} and \ref{alg:line:rec:8} solve a star-shaped problem with at most ${degree}(j)$ variables, which, by inspecting \eqref{eq:star_solution}, we know can be done in $\mathcal{O}(\text{degree}(j))$ steps. Since, $j$ never takes the same value twice, the overall complexity is bounded by $\mathcal{O}(\sum_{j \in \V} \text{degree}(j)) = \mathcal{O}(|\E|)= \mathcal{O}(q)$. The $\mathcal{O}(q)$ bound on memory is possible because all the variables that occupy significant memory are being passed by reference, and are modified in place during the whole recursive procedure.

The following lemma shows how the recursive procedure to solve 
a $(T,\B,\alpha,\beta,\gamma)$-problem can be used to 
compute the rates of change of $Z^*(t)$
of a $(T,\B)$-problem. Its proof follows from the observation that the rate of change of the solution with $t$ in \eqref{eq:star_solution} in Lemma \ref{th:star_problem} only depends on $\alpha$ and $\beta$, and that the reduction equations \eqref{eq:reduction_equations} in Lemma \ref{th:tree_reduction} never make $\alpha'$ or $\gamma'$ depend on $\beta$. 

\begin{lemma}[Rates only]\label{th:rates_only}
Let $Z^*(t)$ be the solution of the $(T,\B)$-problem, and let
 $\tilde{Z}^*(t)$ be the solution of the $(T,\B,{\bf 1},0,{\bf 1})$-problem. Then,
$Z^*(t) = c_1 t + c_2$, and $\tilde{Z}^*(t) = c_1 t$ for some $c_1$ and $c_2$.
\end{lemma}

\begin{figure}[t!]
\includegraphics[width=0.35\columnwidth,trim={0cm -1.5cm 0cm 0cm},clip]{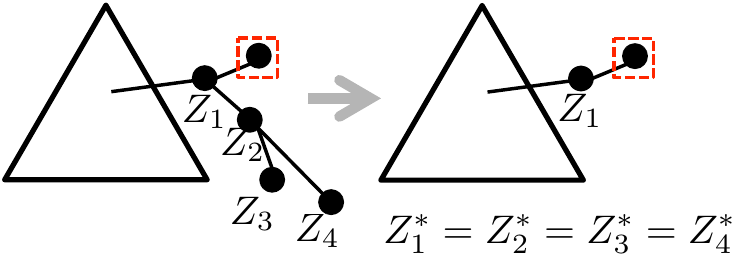}
\includegraphics[width=0.65\columnwidth,trim={0cm 0cm 0cm 0cm},clip]{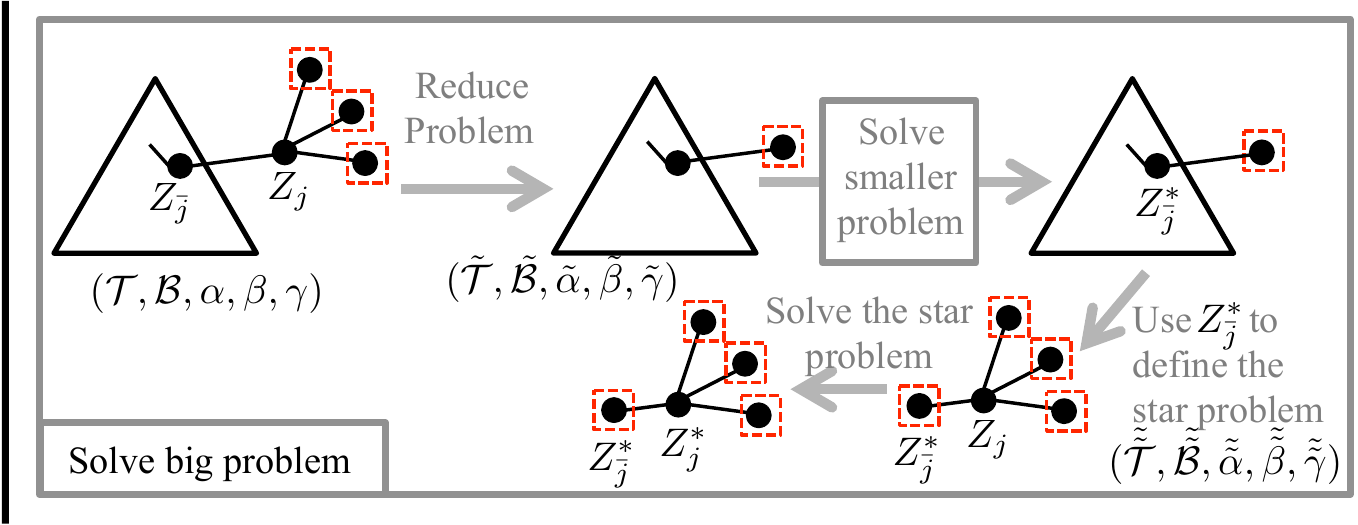}
\caption{\small Red squares represent fixed nodes, and black circles free nodes. (Left) By repeatedly invoking Lemma \ref{th:pruning}, we can remove nodes $2$, $3$, and $4$ from the original problem, since their associated optimal values are equal to the optimal value for node $1$. (Right) We can compute the rates for all the free nodes of a subtree recursively by applying Lemma \ref{th:tree_reduction} and Lemma \ref{th:star_problem}. We
know the linear behavior of variables associated to red squares.}
\vspace{-0.5cm}
 \label{fig:recursion}
\end{figure}

We finally present the full algorithm to compute $Z'^*(t_i)$ and $\LL''*(t_i)$ from 
$T$ and $\B(t_i)$.
\begin{algorithm}[h!]
\caption{ComputeRates (input: $T$ and $\B(t_i)$ output: $Z'^*(t_i)$ and $\LL''(t_i)$)} \label{alg:compute_rates}
\begin{algorithmic}[1]
	\State $Z'^*(t_i)_j = 1$ for all $j \in \B(t_i)$
	\For{ each $(T_w,\B_w)$-problem induced by $\B(t_i)$ }
	\State Set $\tilde{T}_w$ to be $T_w$ pruned of all leaf nodes in $\B_w$, by repeatedly evoking Lemma \ref{th:pruning}
\State $ \text{ComputeRatesRec}(\tilde{T}_w,j,\B_w,{\bf 1},{\bf 0},{\bf 1},\tilde{Z'^*})$\label{alg:line:rates:4}
    \State $Z'^*(t_i)_j = \tilde{Z'^*}_j$ for all $j \in  V_w \backslash \B$ \label{alg:line:rates:5}
    \EndFor
    	\State Compute $\LL''(t_i)$ from $Z'^*(t_i)$ using Lemma \ref{th:computing_ddL_from_dZ}
        \State \textbf{return} $Z'^*(t_i)$ and $\LL''(t_i)$
\end{algorithmic}
\end{algorithm}

\vspace{-0.4cm}
The following theorem follows almost directly from Theorem \ref{th:complexity_of_rec_reduce}. 
\begin{theorem}\label{th:compute_rates_main_theorem}
Alg. \ref{alg:compute_rates} correctly computes $Z'^*(t_i)$ and $\LL''(t_i)$ in $\mathcal{O}(q)$ steps, and uses $\mathcal{O}(q)$ memory.
\end{theorem}
%

%
%
%%%%%%%%%%%%%%%%%%%%%%%%%%%%%%%%
%
%
\vspace{-0.4cm}
\section{Reducing computation time in practice} \label{sec:improvements}

Our numerical results are obtained for an improved  version of Algorithm \ref{alg:projection}. We now explain the main idea behind this algorithm.

The bulk of the complexity of Alg. \ref{alg:projection} comes from line \ref{alg:line:main_alg_compute_rates}, i.e., computing the rates $\{Z'^*(t_i)_j\}_{j \in \V \backslash \B(t_i)}$ from $\B(t_i)$ and $T$. For a fixed $j \in \V \backslash \B(t_i)$, and by Lemma \ref{eq:decomposition_of_problem}, the rate $Z'^*(t_i)_j$, depends only on one particular $(T_w=(r_w,\V_w,\E_w),\B_w)$-problem induced by $\B(t_i)$.
If exactly this same problem is induced by both $\B(t_i)$ and $\B(t_{i+1})$, which happens if the new nodes that become fixed in line \ref{alg:line:main_alg_new_fixed_nodes} of round $i$ of Algorithm \ref{alg:projection} are not in $\V_w \backslash \B_w$, then we can save computation time in round $i+1$, by not recomputing  any rates for $j \in \V_w \backslash \B_w$, and using for $Z'^*(t_{i+1})_j$ the value 
$Z'^*(t_{i})_j$.

Furthermore, if only a few $\{Z'^*_j\}$ change from round $i$ to round $i+1$, then we can also save computation time in computing $\LL''$ from $Z'^*$
by subtracting from the sum in the right hand side of equation  \eqref{eq:computing_ddL_from_dZ} the terms that depend on the previous, now changed, rates, and adding new terms that depend on the new rates.

Finally, if the rate $Z'^*_j$ does not change, then the value of $t < t_i$ at which $Z^*_j(t)$ might intersect $t - N_j$, and become fixed, given by $P_j$ in line \ref{alg:line:next_critical_value}, also does not change. (Note that this is not obvious from the formula for $P_r$ in line \ref{alg:line:next_critical_value}). If not all $\{P_r\}$ change from round $i$ to round $i+1$, we can also save computation time in computing the maximum, and maximizers, in line \ref{alg:line:main_alg_new_fixed_nodes} by storing $P$ in a maximum binary heap, and executing lines \ref{alg:line:next_critical_value} and \ref{alg:line:main_alg_new_fixed_nodes} by extracting all the maximal values from the top of the heap. Each time any $P_r$ changes, the heap needs to be updated.

%
%
%%%%%%%%%%%%%%%%%%%%%%%%%%%%%%%%
%
%

\vspace{-0.5cm}
\section{Numerical results}\label{sec:num_res}
\vspace{-0.3cm}

Our algorithm to solve \eqref{eq:projection_onto_PPM} exactly in a finite number of steps is of
interest in itself. Still, it is  interesting to compare it with other algorithms. In particular, we compare the convergence rate of our algorithm with two popular methods that solve \eqref{eq:projection_onto_PPM} iteratively: the Alternating Direction Method of Multipliers (ADMM), and the Projected Gradient Descent (PGD) method. We apply the ADMM, and the PGD, to both the primal formulation \eqref{eq:projection_onto_PPM}, and the dual formulation \eqref{eq:dual_higher}.
We implemented all the algorithms in C, and derived closed-form updates for ADMM and PG, see Appendix \ref{app:sec:details_of_ADMM_and_PDG}.
We ran all algorithms on a single core of an Intel Core i5 2.5GHz processor. 

Figure \ref{fig:run_time_diff_alg}-(left) compares different algorithms for a random Galton–Watson input tree truncated to have $q=1000$ nodes, with the number of children of each node chosen uniformly within a fixed range, and for a random input $\hat{F}\in\mathbb{R}^q$, with entries chosen i.i.d. from a normal distribution. We observe the same behavior for all random instances that was tested.
We gave ADMM and PGD an advantage by optimally tuning them for each individual problem-instance tested. In contrast, our algorithm requires no tuning, which is a clear advantage. 
At each iteration, the error is measured as $\max_j \{|M_j - M^*_j|\}$. Our algorithm is about $74\times$ faster than its closest competitor (PGD-primal) for $10^{-3}$ accuracy.
In Figure \ref{fig:run_time_diff_alg}-(right), we show the average run time of our algorithm versus the problem size, for random inputs of the same form. The scaling of our algorithm is (almost) linear, and much faster than our $\mathcal{O}(q^2p)$, $p=1$, theoretical bound.
\vspace{-0.3cm}

\begin{figure}[h!] \label{fig:run_time_diff_alg}
\centering
\includegraphics[width=0.49\columnwidth,trim={1.0cm 11cm 2cm 1.3cm},clip]{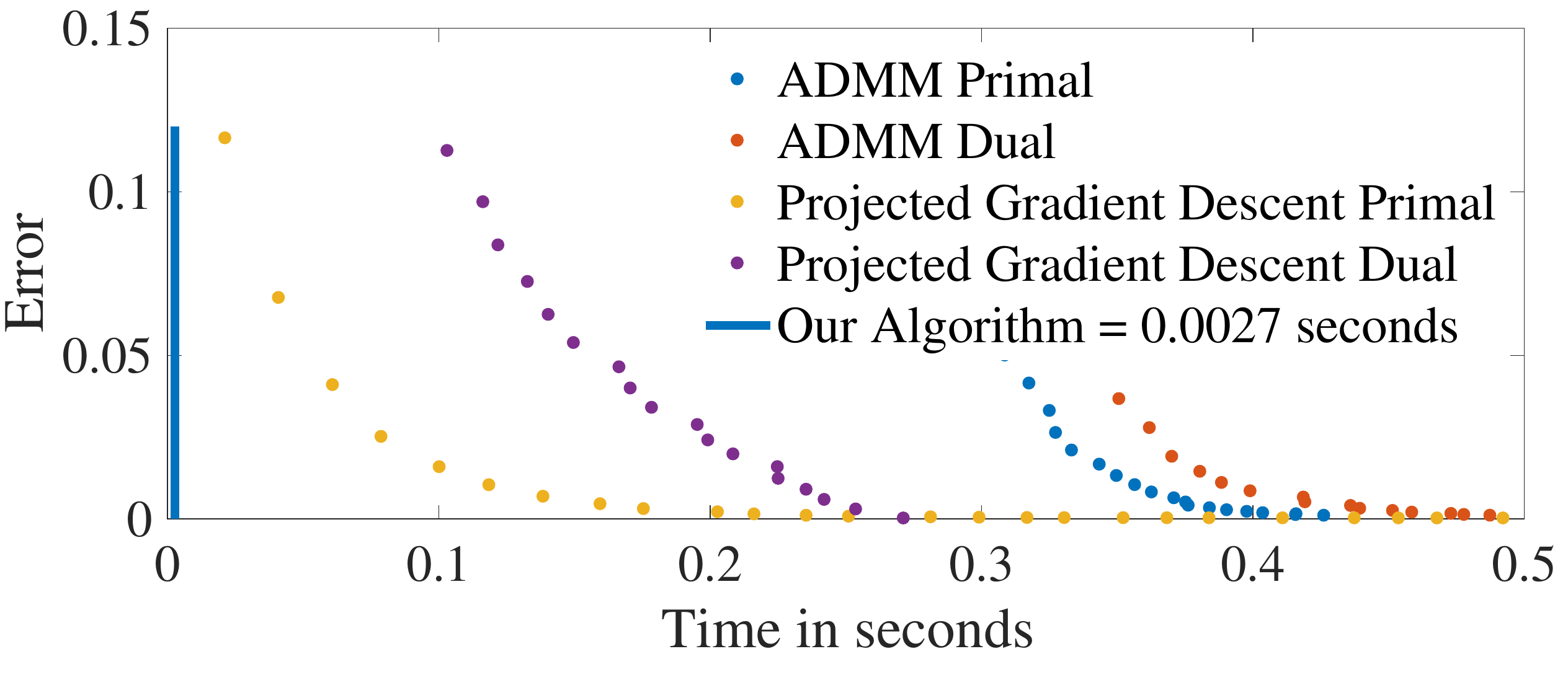}
\;
\includegraphics[width=0.49\columnwidth,trim={0.8cm 9.5cm 1.3cm 10.1cm},clip]{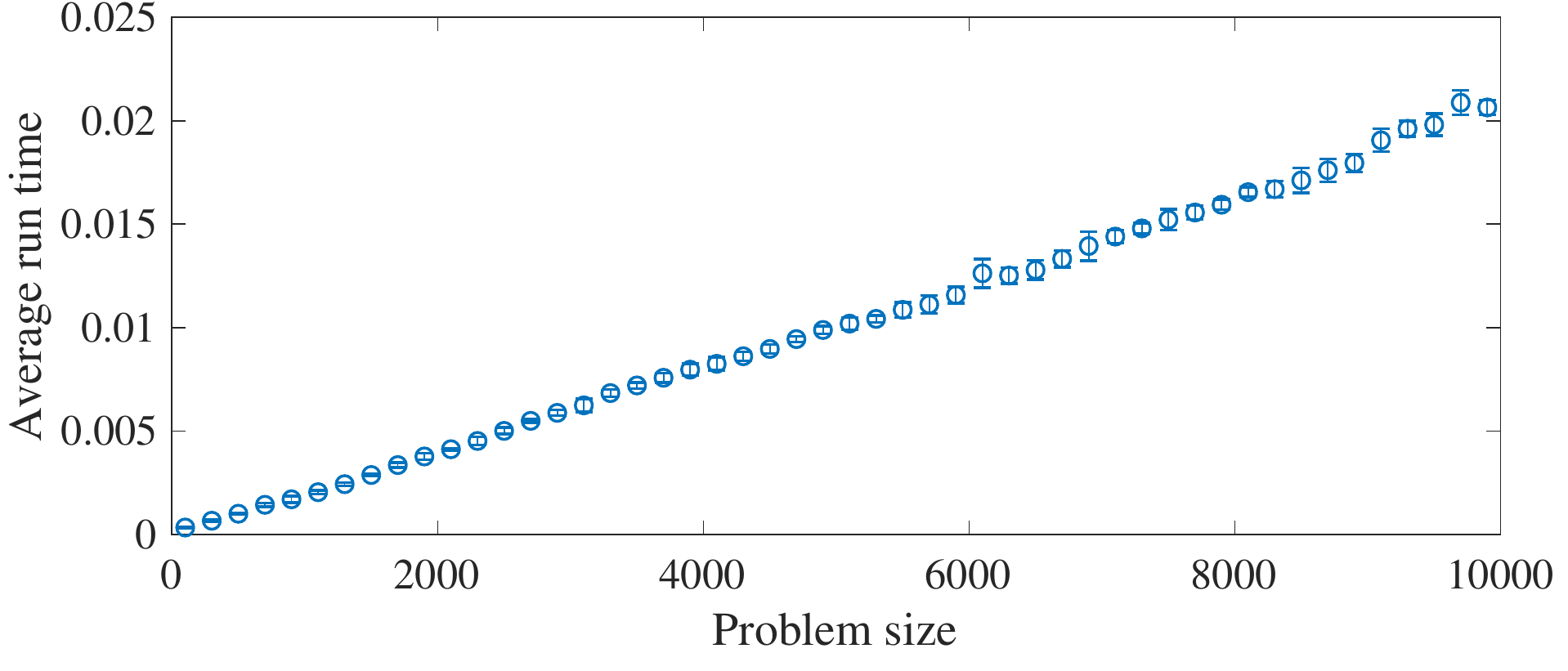}
\caption{\small (Left) Time that the different algorithms take to solve our problem for trees of with $1000$ nodes. (Right) Average run time of our algorithm for problems of different sizes. For each size, each point is averaged over $500$ random problem instances.}
\vspace{-0.35cm}
\end{figure}

Finally, we use our algorithm to exactly solve \eqref{eq:projection_onto_PPM_with_tree}
by computing $\mathcal{C}(U)$ for all trees and a given input $\hat{F}$. Exactly solving \eqref{eq:projection_onto_PPM_with_tree} is very important for biology, since several relevant phylogenetic tree inference problems deal with trees of small sizes. We use an NVIDIA QUAD P5000 GPU to compute the cost of all possible trees with $q$ nodes in parallel, and return the tree with the smallest cost. Basically, we assign to each GPU virtual thread a unique tree, using Prufer sequences \cite{prufer1918neuer}, and then have each thread compute the cost for its tree.
For $q= 10$, we compute the cost of all $100$ million trees in about $8$ minutes, and for $q = 11$, we compute the cost of all $2.5$ billion trees in slightly less than $2.5$ hours.
%The GPU limits our tests to $q = 13$, for which all $1.8$ trillion trees take about $1$ month to scan. 

Code to solve \eqref{eq:projection_onto_PPM} using Alg. \ref{alg:projection}, with the improvements of Section \ref{sec:improvements}, can be found in \cite{gitcode}. More results using our algorithm can be found in Appendix \ref{sec:app:more_results}.

%
%
%%%%%%%%%%%%%%%%%%%%%%
%
%
\vspace{-0.4cm}
\section{Conclusions and future work}
\vspace{-0.2cm}

We propose a new direct algorithm that, for a given tree,
computes how close the matrix of frequency of mutations per position is to satisfying the perfect phylogeny model. Our algorithm is faster than the state-of-the-art iterative methods for the same problem, even if we optimally tune them.
We use the proposed algorithm to build a GPU-based phylogenetic tree inference engine for the trees of relevant biological sizes. Unlike existing algorithms, which 
%one does not need to 
only heuristically search a small part of the space of possible trees, our algorithm performs a complete search over all trees relatively fast.
It is an open problem to find direct algorithms that can provably solve our problem
in linear time on average, or even for a worst-case input.
%It is an interesting open question to determine if there are direct algorithms that can solve our problem
%in linear time on average, or even for a worst-case input.

{\bf Acknowledgement: } This work was partially funded by NIH/1U01AI124302, NSF/IIS-1741129, and a NVIDIA hardware grant.

%
%
%%%%%%%%%%%%%%%%%%%%%%
%
%

\bibliographystyle{unsrt}
\bibliography{references}

%
%
%%%%%%%%%%%%%%%%%%%%%%
%
%

\newpage
\clearpage

\appendix
{\bf \Large Appendix for ``Efficient Projection onto the Perfect Phylogeny Model''}

\section{Further illustrations}\label{sec:app:further_illustrations}

\begin{figure}[h!] 
\centering
\includegraphics[width=4.5cm]{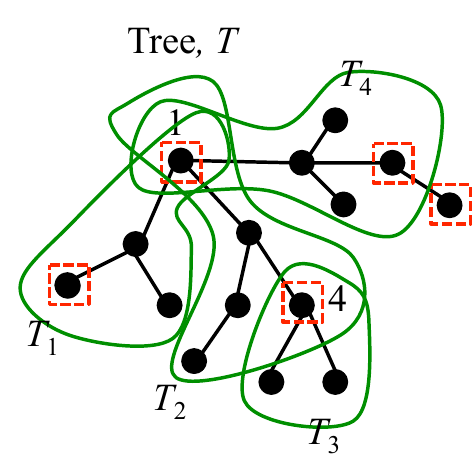}
\caption{\small Four subtrees of $T$ induced by $\B(t)$, represented by the red squares. The root of $T_1$, $T_2$ and $T_4$ is node $1$. The root of $T_3$ is node $4$. All subtrees must have  nodes associated to free variables (free nodes). Any subtree is uniquely identified by any free node in it. Within each subtree, any fixed node must be the root or a leaf.}
\label{fig:tree_and_variables}
\end{figure}

\section{Proof of Theorem \ref{th:dual_prob} in Section \ref{sec:main_results}} \label{sec:app:proof_of_th:dual_prob}

We prove Theorem \ref{th:dual_prob}, by first proving the following very similar theorem.

\begin{theorem} \label{th:compact_dual_prob}
Problem \eqref{eq:projection_onto_PPM} can be solved by solving
\begin{align} 
& \min_{t} t + \mathcal{L}(t),\label{eq:compact_dual_higher}\\[-0.4cm]
& \hspace{1.2cm}\mathcal{L}(t) = \min_{Z\in\mathbb{R}^q} \frac{1}{2}\|(U^\top)^{-1} Z \|^2 \text{ subject to } Z + N \leq t \ones \label{eq:compact_dual},
\end{align}
where $N = U^\top \hat{F}$.
In particular, if $t^*$ minimizes 
\eqref{eq:compact_dual_higher}, 
$Z^*$ minimizes \eqref{eq:compact_dual} for $t=t^*$, and $M^*,F^*$ minimize \eqref{eq:projection_onto_PPM}, then
\begin{equation}\label{eq:compact_Z_star_M_star_relation}
M^* = - U^{-1}(U^{-1})^\top Z^*, \quad F^* = -(U^{-1})^\top Z^*.
\end{equation}
Furthermore, $t^*$, $M^*$, $F^*$ and $Z^*$ are unique.
\end{theorem}

\begin{proof}[Proof of Theorem \ref{th:compact_dual_prob}]
Problem \eqref{eq:projection_onto_PPM} depends on the tree $T$ through the matrix of the ancestors, $U$. 
To see how Theorem \ref{th:compact_dual_prob} implies Theorem \ref{th:dual_prob}, it is convenient to make this dependency more explicit.
Any tree in $\T$, can be represented through a binary matrix $T$, where $T_{ij}=1$ if and only if node $i$ is the closest ancestor of node $j$. 
Henceforth, let $\T$ denote the set of all such binary matrices.
We need the following lemma, which we prove later in this section of the appendix.
\begin{lemma}\label{th:bijection_U_T}
Consider an evolutionary tree and its matrices $T \in \T$ and  $U \in \U$. We have  
\begin{equation}\label{eq:U_and_T_relation}
U = (I - T)^{-1}.
\end{equation}
\end{lemma}

Eq.~\eqref{eq:U_and_T_relation}
implies that
$((U^{-1})^\top Z)_i = (Z - T^\top Z)_i
 = Z_i - Z_{\bar{i}}$, and that $U^{-1}((U^{-1})^\top Z)_i = Z_i - Z_{\bar{i}} - \sum_{r \in \partial i} (Z_r - Z_{\bar{r}})$, where $\partial i$ denotes the children of $i$ in $T$, 
 $\bar{i}$ represents the closest ancestor of $i$ in $T$. We assume by convention that $Z_{\bar{i}} = 0$ when $i = r$ is the root of $T$.
Furthermore, the definition of $U$ implies that $N_i = (U^\top \hat{F})_i = \sum_{j \in \Delta i} \hat{F}_j$, where $\Delta i$ denotes the ancestors of $j$.
Thus,
\begin{align} \label{eq:mechanical_interp}
\mathcal{L}(t) = &  \min_{Z\in\mathbb{R}^q} \frac{1}{2}\sum_{i \in \V} (Z_i - Z_{\bar{i}})^2 \text{ subject to }\\
&  Z_i \leq t - \sum_{j \in \Delta i} \hat{F}_j \;,\forall i\in \V,\nonumber\\
& M^*_i = -Z^*_i + Z^*_{\bar{i}} + \sum_{r \in \partial i} (Z^*_r - Z^*_{\bar{r}}) \text{ and } \\
& F^*_i = -Z^*_i + Z^*_{\bar{i}}, \forall i\in \V \nonumber. \label{eq:fast_computation_of_M_and_F_from_Z}
\end{align}
\end{proof}

\begin{proof}[Proof of Theorem \ref{th:dual_prob}]
Our proof is based on Moreau's decomposition \cite{parikh-boyd}. Before we proceed with the proof, let us introduce a few concepts.

Given a convex, closed and proper function $g:\mathbb{R}^q \mapsto \mathbb{R}$, we define its proximal operator by the map $G: \mathbb{R}^q \mapsto \mathbb{R}^q$ such that
\begin{equation}\label{eq:def_PO_g}
G(n) = \arg \min_{x\in \mathbb{R}^q} g(x) + \frac{1}{2}\|x - n\|^2,
\end{equation}
where in our case $\|\cdot\|$ is the Euclidean norm. We define the Fenchel dual of $g$ as
\begin{equation}
g^*(x) = \sup_{s\in\mathbb{R}^q} \{x^\top s - g(s)\},
\end{equation}
and we denote the proximal operator of $g^*$ by $G^*$. Note that $G^*$ can be computed from definition \eqref{eq:def_PO_g} by replacing $g$ by $g^*$.

Moreau's decomposition identity states that
\begin{equation}
G(n) + G^*(n) = n.
\end{equation}

We can now start the proof.
Consider the following indicator function
\begin{equation}\label{eq:def_of_g_for_our_prob}
g(\tilde{M}) =  
\begin{cases}
0, \qquad \text{if } (U^{-1}\tilde{M}) \geq 0 \ \text{and} \ \mathbf{1}^{\text{T}} (U^{-1}\tilde{M}) = 1, \\
+\infty, ~~\ \text{otherwise},
\end{cases}
\end{equation}
where $\tilde{M} \in \mathbb{R}^q$, and consider its associated proximal operator $G$.
Solving problem \eqref{eq:projection_onto_PPM}, i.e., finding a minimizer $M^*$, is equivalent to evaluating $U^{-1} G(\hat{F})$. Using Moreau's decomposition, we have
\begin{equation}
M^* = U^{-1} G(\hat{F}) = U^{-1}\hat{F} -  U^{-1}G^*(\hat{F}).
\end{equation}
We will show that $G^*(\hat{F}) = \hat{F} + (U^{-1})^\top Z^*$, where $Z^*$ is a minimizer of \eqref{eq:dual}, which proves \eqref{eq:Z_star_M_star_relation} and essentially completes the proof.

To compute $G^*$, we first need to compute
\begin{align}
g^*(Y) &= \sup_{\tilde{M}} \{Y^\top \tilde{M} - g(\tilde{M})\}\\
&=\max_{\tilde{M}} Y^\top \tilde{M}\label{eq:g_start_middle}\\
& \hspace{0.5cm} \text{ subject to }~~U^{-1} \tilde{M} \geq 0, \ones^\top(U^{-1} \tilde{M})=1\nonumber.
\end{align}
Making the change of variable $M = U^{-1} \tilde{M}$, the maximum in problem \eqref{eq:g_start_middle} can be re-written as 
\begin{align}
&\max_{{M}} (U^\top Y)^\top  M \label{eq:g_start_middle_2}\\
& \hspace{0.cm} \text{ subject to }~~{M} \geq 0, \ones^\top {M}=1\nonumber.
\end{align}
It is immediate to see that
the maximum in \eqref{eq:g_start_middle_2} is achieved 
if we set all components of $M$ equal to zero except the one corresponding to the largest component of the vector $U^\top Y$, which we should set to one.
Therefore, we have 
\begin{equation}
g^*(Y) = \max_i \; (U^\top Y)_i.
\end{equation}
Now we can write
\begin{align}\label{eq:appendix_towards_uniqueness}
G^*(\hat{F}) &= \arg \min_{Y\in \mathbb{R}^q}\; g^*(Y) + \frac{1}{2}\|Y - \hat{F}\|^2\\
&=\arg \min_{Y\in \mathbb{R}^q,t\in\mathbb{R}} \;t + \frac{1}{2}\|Y - \hat{F}\|^2\\
& \hspace{1.3cm} \text{ subject to }~~U^\top Y \leq t\nonumber.
\end{align}

Making the change of variable $Z = U^\top (Y -\hat{F})$, we can write $G^*(\hat{F})$ as
\begin{align}
G^*(\hat{F}) 
&=\hat{F} + (U^{-1})^\top Z^*, \text{ where } \\
(Z^*, t^*) &= \arg \min_{Z\in \mathbb{R}^q,t\in\mathbb{R}} \;t + \frac{1}{2}\|(U^{-1})^\top Z\|^2\\
& \hspace{1cm} \text{ subject to }~~Z + U^\top \hat{F} \leq t\nonumber.
\end{align}

To see that $M^*$ and $F^*$ are unique, notice that 
problem \eqref{eq:projection_onto_PPM} is a projection onto a convex set polytope, which always has a unique minimizer.
Moureau's decomposition implies that $G^*(\hat{F})$ is unique, hence the minimizer $Y^*$ of \eqref{eq:appendix_towards_uniqueness} is unique. Thus, $Z^* = U^\top(Y^* - \hat{F})$ and $t^* = g^*(Y^*)$ are also unique.
\end{proof}

\begin{proof}[Proof of Lemma \ref{th:bijection_U_T}]
We assume that the tree has $q$ nodes.
The matrix $T$ is such that $T_{v,v'}=1$ if and only if $v$ is the closet ancestor of $v'$. Because of this, the $v$th column of $T^k$ has a one in row $v'$ if and only if $v'$ is an ancestor $v$ separated by $k$ generations. Thus, the $v$th column of~$I + T + T^2 +\dots+ T^{q-1}$, contains a one in all the rows $v'$ such that $v'$ is an ancestor of $v$, or if $v=v'$. But this is the definition of the matrix $U$ associated to the tree $T$.
Since no two mutants can be separated by more than $q-1$ generations, $T^k = 0$ for all $k \geq q$. It follows that 
$$U = I + T + T^2 +\dots +T^{q-1} = \sum^\infty_{i=0} T^i = (I-T)^{-1}.$$
\end{proof}
%
%
%
%
%%%%%%%%%%%%%%%%%%%%%%%%%%%%%%%%%%%%%%%%%%%%%%%%%%%
%
%
\section{Proof of useful observations in Section \ref{sec:useful_obs}} \label{sec:app:proof_of_useful_obser}
\begin{proof}[Proof of Lemma \ref{th:convexity_of_L_dual}]
The proof follows from the following generic fact, which we prove first. Let 
$g(W) = \min_{Z\in \mathbb{R}^q} f(Z,W)$.  If $f$ is convex in $(Z,W)$, then $g$ is convex.

Indeed, let $\alpha\geq 0$ and $\alpha' =  1 - \alpha$. We get $\alpha g( W_1)  + \alpha'g(W_2) = \min_{Z_1,Z_2} \alpha f(Z_1,W_1) + \alpha' f(Z_2,W_2) \geq  \min_{Z_1,Z_2}  f(\alpha Z_1 + \alpha' Z_2,\alpha W_1 + \alpha' W_2) = g(\alpha W_1 + \alpha' W_2)$.

To apply this result to our problem, let $f_1(Z)$ be the objective of \eqref{eq:mechanical_interp} and let $f_2(Z,t,N)$ be a function (on the extended reals) such that $f_2 = 0 $ if $(Z,t,N)$ satisfy the constraints in \eqref{eq:mechanical_interp} and $+\infty$ otherwise. Now notice that $\LL(t) = \min_{Z} f_1(Z) + f_2(Z,t)$, where $f_1 + f_2$ is convex in $(Z,t,N)$, since both $f_1$ and $f_2$ are convex in $(Z,t,N)$.
Convexity implies that $\LL$ is continuous in $N$ and $t$. It also implies that $\LL'(t)$ is non increasing~in~$t$.
\end{proof}
\begin{proof}[Proof of Lemma \ref{th:continuity_of_Z_t}]
{Continuity of $Z^*(t)$: }
The objective function in \eqref{eq:mechanical_interp} is convex as a function of $Z$ and has unique minimum at $Z_i = 0, \forall i$. Hence, it is strictly convex. Due to strict convexity, if the objective takes values in a small interval, then $Z$ must be inside some small ball.

Since we know, by the remark following Lemma \ref{th:convexity_of_L_dual}, that $\LL$ is continuous as a function of $t$, if $t$ and $t'$ are close, then $\LL(t)$ and $\LL(t')$ must be close.  Strict convexity then implies that $Z^*(t)$ and $Z^*(t')$ must be close.
The same argument can be used to prove continuity with respect to $N$.

{Continuity of $Z^*(t^*)$: }
Recall that $Z^*(t^*) = Z^*$, the solution of 
\eqref{eq:projection_onto_PPM}.
$Z^*$ is a continuous function of $M^*$, which
is the solution to \eqref{eq:projection_onto_PPM},
and thus is fully determined by $U$ and $\hat{F}$.
Since, $\hat{F} = (U^\top)^{-1} N$, $\hat{F}$ is a continuous function of $N$, and it suffices to prove that $M^*$ is continuous in $\hat{F}$.
Problem \eqref{eq:projection_onto_PPM} finds the projection of $\hat{F}$
onto a convex polytope. Let $F^*$ be this projection.
Since $F^*$ changes continuously with $\hat{F}$, $M^* = U^{-1} F^*$
also changes continuously with $\hat{F}$.
\end{proof}
\begin{proof}[Proof of Lemma \ref{th:B_piece_wise_constant}]
Since $Z^*(t)$ is continuous, if $Z^*(t)_i \neq t - N_i$ then $Z^*(t')_i \neq t' - N_i$ for $t'$ in some neighborhood of $t$.
\end{proof}
\begin{proof}[Proof of Lemma \ref{eq:decomposition_of_problem}]
First note that, by definition of $\B(t)$, we know the value of all variables in $\B(t)$. Hence, the unknowns in problem \eqref{eq:mechanical_interp} are the variables in $\V \backslash \B(t)$, which can be partitioned into disjoint sets $\{\V_i \backslash \B(t)\}^k_{i=1}$.

Second notice that for each term in the objective \eqref{eq:mechanical_interp} that involves not known variables, there is some subtree $T_i$ that contains both of its variables. 
It follows that, given $\B(t)$, problem \eqref{eq:mechanical_interp} breaks into $k$ independent problems, the $i$th problem having as unknowns only the variables in $\V_i \backslash \B(t)$ and all terms in the objective where either $j$ or $\bar{j}$ are in $\V_i \backslash \B(t)$. 

Obviously, if $j \in \mathcal{V}_w  \cap \B(t)$, then, by definition, $Z^*(t)_j = c_1 t + c_2$, with $c_1 = 1$.
To find the behavior of $Z^*(t)_j$ for $j \in \mathcal{V}_w  \backslash \B(t)$, we need to solve \ref{eq:simpler_sub_problem}.
To solve \eqref{eq:simpler_sub_problem}, notice that the first-order optimality conditions for problem \eqref{eq:mechanical_interp}
imply that, if $j \in \V \backslash \B(t)$, then
\begin{equation}
Z_j = \frac{1}{|\partial j|} \sum_{r \in \partial j } Z_r,
\end{equation}
where $\partial j$ denotes the neighbors of node $j$. We can further write
\begin{align}
&Z_j = \frac{1}{|\partial j|} \sum_{r \in \partial j \cap \B(t)} Z_r + \frac{1}{|\partial j|} \sum_{r \in \partial j \backslash \B(t)}, \nonumber\\
& Z_r
=\frac{1}{|\partial j|} \sum_{r \in \partial j \cap \B(t)} (t - N_r) + \frac{1}{|\partial j|} \sum_{r \in \partial j \backslash \B(t)} Z_r\label{eq:sol_for_Z_j}.
\end{align}
It follows that 
$Z_j = c_1t + c_2$, for some $c_1$ and $c_2$ that depend on $T$, $N$ and $\B$.
If we solve for $Z_j$ by recursively applying 
\eqref{eq:sol_for_Z_j}, it is immediate to see that $c_1 \geq 0$.

To see that $c_1\leq 1$, we study how $Z_j$, defined by \eqref{eq:sol_for_Z_j}, depends on $t$ algebraically. To do so, we treat $t$ as a variable. The study of this algebraic dependency in the proof should not be confused with $t$ being fixed in the statement of the theorem.

Define $\rho = |\partial i \cup \B(t)|/|\partial j|$, and notice that 
\begin{align}
\max_j\{ Z_j\} \leq \rho t + (1-\rho) \max_j \{Z_j\} + C,
\end{align}
in which $C$ is some constant. Recursively applying the above inequality we get 
\begin{equation}
\max_j\{ Z_j\} \leq t + C',
\end{equation}
in which $C'$ is some constant.
This shows that no $Z_j$ can grow with $t$ faster than $1 \times t$ and hence $c_1 \leq 1$.
\end{proof}

\begin{proof}[Proof of Lemma \ref{th:Z_and_Lprime_are_piece_wise_linear}]
Lemma \ref{eq:decomposition_of_problem} implies that, for any $j$, $Z^*_j(t)$ depends linearly on $t$. The particular linear dependency, depends on $\B(t)$, which is piecewise constant by Lemma \ref{th:B_piece_wise_constant}. Therefore, $Z^*_j(t)$ is a continuous piecewise linear function of $t$. This in turn implies that $\LL'(t)$ is a continuous piecewise linear function of $t$, since it is the derivative of the 
continuous piecewise quadratic $\LL(t) = (1/2)\sum_{i \in \V} (Z^*(t)_i - Z^*(t)_{\bar{i}})^2$. Finally, since the particular  linear dependency of $Z^*$, depends on $\B(t)$,
it follows that $Z^*(t)$ and $\LL'(t)$ change linear segment if and only if $\B(t)$ changes.
\end{proof}
\begin{proof}[Proof of Lemma \ref{eq:B_always_grows}]
Let us assume that there exists $t < t'$ for which $\B(t) \subset \B(t')$. We can assume without loss of generality that $t$ is sufficiently close to $t'$ such that $\B(s)$ is constant for $s\in [t,t')$. 
Let $j$ be such that $j\in B(t')$ but $j\notin B(t)$. This means that $Z^*_j(s) < s - N_j$ for all $s\in [t,t')$ and that
$Z^*_j(t') = t' - N_j$. Since by Lemma \ref{eq:decomposition_of_problem}, $Z^*_j(s) = c_1 s + c_2$, for some constants $c_1$ and $c_2$, the only way that $Z^*_j(s)$ can intersect $s - N_j$ at $s = t'$ is for 
$c_1 > 1$, which is a contradiction.

If $\B(t)$ decreases as $t$ increase, and given that the largest that $\B(t)$ can be is $\{1,\dots,q\}$, it follows that $\B(t)$ can only take $q+1$ different configurations. One configuration per size of $\B(t)$, from $q$~to~$0$.
\end{proof}
\begin{proof}[Proof of Lemma \ref{th:Z_and_Lprime_finite_break_points}]
Lemma \ref{eq:B_always_grows} implies that
$\B(t)$ changes at most $q+1$ times. Lemma \ref{th:Z_and_Lprime_are_piece_wise_linear} then implies that $Z^*(t)$ and $\LL'(t)$ have less than $q+1$ different linear segments.
\end{proof}

%
%
%
%%%%%%%%%%%%%%%%%%%%%%%%%%%%%%%%%%%%%%%%%%%%%%%%%%%%
%
%

\section{Proofs of the properties of the algorithm in Section \ref{sec:main_alg}}

\begin{proof}[Proof of Theorem \ref{th:main_alg_complexity}]
{\bf{Run-time:}} Recall that $Z^*, Z'^* \in \mathbb{R}^q$ and that $\LL' \in \mathbb{R}$.
Line \ref{alg:line:compute_N_from_F} is done in $\mathcal{O}(q)$ steps by doing a DFS on $T$. Here, we assume that $T$ is represented as a linked list. Specifically, starting from the root, we keep a variable $x$ where we accumulate the values of $\hat{F}_j$ visited from the root to the current node being explored in $T$ as we move down the tree. As we move up the tree, we subtract values of the nodes $\hat{F}_j$ from $x$. Then, at each node $i$ visited by the DFS, we can read from $x$ the value $N_i$.
Line \ref{alg:line:init} takes $\mathcal{O}(q)$ steps to finish.
The procedure ComputeRates takes $\mathcal{O}(q)$ steps to finish, which we prove in Theorem \ref{th:compute_rates_main_theorem}. All of the other lines inside the for-loop are manipulations that take at most $\mathcal{O}(q)$ steps.
Lines \ref{alg:line:find_Zstar} and \ref{alg:line:find_tstar} take $\mathcal{O}(q)$ steps.
From \eqref{eq:Z_star_M_star_relation}, the complexity to compute $F^*$ is $\mathcal{O}(q)$, and the complexity to compute $M^*$ is $\mathcal{O}(\sum_{i \in \V} |\partial i|) = \mathcal{O}(|\E|) = \mathcal{O}(q)$.

{\bf{Memory: }} The DFS in line \ref{alg:line:compute_N_from_F} only requires $\mathcal{O}(q)$  memory.
Throughout the algorithm, we only need to keep the two most recent values of $t_i$, $\B(t_i)$, $Z^*(t_i)$, $Z'^*(t_i)$, $\LL'(t_i)$ and $\LL''(t_i)$. This takes $\mathcal{O}(q)$ memory. The procedure ComputeRates takes $\mathcal{O}(q)$ memory, which we prove in Theorem \ref{th:compute_rates_main_theorem}.
\end{proof}
\begin{proof}[Proof of Theorem \ref{th:alg_correctness}]

The proof of Theorem \ref{th:alg_correctness} amounts to checking that, at every step of Algorithm \ref{alg:projection}, the quantities computed, e.g., the paths $\{Z^*(t)\}$ and $\{\LL'(t)\}$, are correct.

Lemmas \ref{th:Z_and_Lprime_are_piece_wise_linear} and \ref{th:Z_and_Lprime_finite_break_points} prove that $Z^*(t)$ and $\LL'(t)$ 
are piecewise linear and continuous with at most $q$ changes in linear segment. Hence, the paths $\{Z^*(t)\}$ and $\{\LL'(t)\}$ are fully specified by their value at $\{t_i\}^k_{i=1}$, and $k \leq q$.

Lemma \ref{th:Z_and_Lprime_are_piece_wise_linear} proves that these critical values are determined as the instants, at which $\B(t)$ changes. Furthermore, Lemma \ref{eq:B_always_grows} proves that, as $t$ decreases, variables are only added to $\B(t)$. Hence, to find $\{t_i\}$ and $\{\B(t_i)\}$, we only need to find the times and components at which, as $t$ decreases, $Z^*(t)_r$ goes from $Z^*(t)_r < t - N_r$ to $Z^*(t)_r = t - N_r$. Also, since $\B$ can have at most $q$ variables, the for-loop in line \ref{alg:line:forloop} being bounded to the range $1$-$q$, does not prevent the algorithm from finding any critical value.

Theorem \ref{th:compute_rates_main_theorem} tells us that we can compute $Z'^*(t_i)$ from $\B(t_i)$ and $T$. Since we have already proved that the path $\{Z^*(t)\}$ is piecewise linear and continuous, we can compute $t_{i+1}$, and the variables that become fixed, by solving \eqref{eq:formula_for_next_t} for $t$ for each $r \notin \B(t_i)$, and choosing for $t_{i+1}$  the largest such $t$, and choosing for the new fixed variables, i.e., $\B(t_{i+1})-\B(t_{i})$, the components $r$ for which the solution of \eqref{eq:formula_for_next_t} is $t_{i+1}$.

Since we have already proved that that $Z^*(t)$ and $\LL'(t)$ 
are piecewise linear and constant, we can compute $Z^*(t_{i+1})$ and $\LL'(t_{i+1})$ from $Z^*(t_{i})$, $\LL'(t_{i})$, $Z'^*(t_{i})$ and $\LL''(t_{i})$ using \eqref{eq:liner_relation_between_critical_points}.

Lemma \ref{th:convexity_of_L_dual} proves that $\LL'(t)$ decreases with $t$, and Theorem \ref{th:dual_prob} proves that $t^*$ is unique. Hence, as $t$ decreases, there is a single $t$ at which $\LL'(t)$ goes from $>-1$ to $< -1$. Since we have already proved that we correctly, and sequentially, compute  $\LL'(t_i)$,  $\LL''(t_i)$, and that $\LL'(t)$ is piecewise linear and constant, we can stop computing critical values whenever we can determine that $\LL'(t) = \LL'(t_k) + (t-t_k) \LL''(t_k)$ will cross the value $-1$, where $t_k$ is the latest computed critical value. This is the case when $\LL'(t_k) > -1$ and $\LL'(t_{k+1}) < -1$, or when $\LL'(t_k) > -1$ and $t_k$ is the last possible critical value, which happens when $|\B(t_i)| = q$.
From this last critical value, $t_k$, we can then find $t^*$ and $Z^*$ by solving $-1  = \LL'(t_k) + (t^*-t_k) \LL''(t_k)$ and $Z^* =  Z^*(t_k) + (t^* - t_k) Z'^*(t_k)$.
Finally, once we have $Z^*$, we can use \eqref{eq:Z_star_M_star_relation}  in Theorem \ref{th:dual_prob} to find $M^*$ and $F^*$.
\end{proof}

%
%
%%%%%%%%%%%%%%%%%%%%%%%%%%%%%%%%%%%%%%%%%%%%%
%
%
\section{Proofs for computing the rates in Section \ref{sec:computing_rates}} \label{sec:app:proofs_for_rates}

\begin{proof}[Proof of Lemma \ref{th:computing_ddL_from_dZ}]
Let $t \in (t_{i+1},t_{i})$. We have,
\begin{align}
&\LL'(t) = \frac{{\rm d} }{{\rm d} t} \frac{1}{2}\sum_{j \in \V} (Z^*(t)_j - Z^*(t)_{\bar{j}})^2= \nonumber\\
& \sum_{i \in \V} (Z^*(t)_j - Z^*(t)_{\bar{j}})({Z^*}'(t)_j - {Z^*}'(t)_{\bar{j}}).
\end{align}
Taking another derivative, and recalling that $Z''^*(t) = 0$ for $t \in (t_{i+1},t_{i})$, we get
\begin{align}
&\LL''(t) = \sum_{j \in \V} ({Z^*}'(t)_j - {Z^*}'(t)_{\bar{j}})^2,
\end{align}
and the lemma follows by taking the limit $t \uparrow t_i$.
\end{proof}

\begin{proof}[Proof of Lemma \ref{th:pruning}]

The $(T,\B,\alpha,\beta,\gamma)$-problem is unconstrained and convex, hence we can solve it by taking derivatives of the objective with respect to the free variables, and setting them to zero. Let us call the objective function $F(Z)$. 
If $j \in \V \backslash \B$ is a leaf, then $\frac{{\rm d} F}{{\rm d} Z_j} = 0$ implies that $Z^*_j = Z^*_{\bar{j}}$. 
We now prove the second part of the lemma. Let $\tilde{F}(Z)$ be the objective of the modified problem. Clearly,  $\frac{{\rm d} F}{{\rm d} Z_i} = \frac{{\rm d} \tilde{F}}{{\rm d} Z_i}$ for all $i \in \tilde{T} \backslash \bar{j}$. 
Let $C$ be the children of $\bar{j}$ in $T$ and $\tilde{C}$ be the children of $\bar{j}$ in $\tilde{T}$. We have $\tilde{C} = C \backslash j$. Furthermore, $\frac{{\rm d} \tilde{F}}{{\rm d} Z_j} = 0$ is equivalent to $\gamma_j (Z_j - Z_{\bar{j}}) + \sum_{s \in \tilde{C}} \gamma_s (Z_j - Z_s) = 0$, and $\frac{{\rm d} F}{{\rm d} Z_j} = 0$ is equivalent to $\gamma_j (Z_j - Z_{\bar{j}})  + \sum_{s \in C} \gamma_s (Z_j - Z_s) = 0$. However, we have already proved that the optimal solution for the original problem has $Z^*_j= Z^*_{\bar{j}}$. Hence, this condition can be replaced in $\frac{{\rm d} F}{{\rm d} Z_j}$, which becomes $\gamma_j (Z_j - Z_{\bar{j}}) + \sum_{s \in \tilde{C}} \gamma_s (Z_j - Z_s) = 0$. Therefore, the two problems have the same optimality conditions, which implies that $Z^*_i = \tilde{Z}^*_i$, for all $i \in \tilde{\V}$.
\end{proof}

\begin{proof}[Proof of Lemma \ref{th:star_problem}]
The proof follows directly from the first order optimality conditions, a linear equation that we solve for $Z^*_1$.
\end{proof}
\begin{proof}[Proof of Lemma \ref{th:tree_reduction}]
The first order optimality conditions for both problems are a system of linear equations, one equation per free node in each problem.
All the equations associated to the ancestral nodes of $j$ are the same for both problems. 
The equation associated to variable $j$ in the $(T,\B,\alpha,\beta,\gamma)$-problem is
\begin{equation}
\gamma_j (Z_{\bar{j}} - Z_j) + \sum^r_{i=1} \gamma_i (Z_i - Z_j) = 0,
\end{equation}
which implies that 
\begin{equation}\label{eq:Z_j_in_equiv_reduction}
Z_j  = \frac{ \gamma_j Z_{\bar{j}} + \sum^r_{i=1} \gamma_i Z_i }{ \gamma_j + \sum^r_{i=1} \gamma_i }.
\end{equation}

The equation associated to the variable $\bar{j}$ in the $(T,\B,\alpha,\beta,\gamma)$-problem is
\begin{equation} \label{eq:Z_j_bar_in_equiv_reduction}
F(Z,\alpha,\beta,\gamma) + \gamma_j (Z_j - Z_{\bar{j}}) = 0,
\end{equation}
where $F(Z)$ is a linear function of $Z$ determined by the tree structure and parameters associated to the ancestral edges and nodes of $\bar{j}$.
The equation associated to the variable $\bar{j}$ in the $(\tilde{T},\tilde{\B},\tilde{\alpha},\tilde{\beta},\tilde{\gamma})$-problem is
\begin{equation}\label{eq:Z_j_bar_in_equiv_reduction_smaller_problem}
F(\tilde{Z},\tilde{\alpha},\tilde{\beta},\tilde{\gamma}) + \tilde{\gamma}_j (\tilde{\alpha}_j t + \tilde{\beta}_j - \tilde{Z}_{\bar{j}}) = 0,
\end{equation}
for the same function $F$ as in \eqref{eq:Z_j_bar_in_equiv_reduction}.
 Note that the components of $\tilde{\alpha}$, $\tilde{\beta}$ and $\tilde{\gamma}$ associated to the ancestral edges and nodes of $\bar{j}$ are the same as in $\alpha$, $\beta$ and $\gamma$. Hence, 
$F(\tilde{Z},\tilde{\alpha},\tilde{\beta},\tilde{\gamma}) = F(\tilde{Z},\alpha,\beta,\gamma)$.
 
 By replacing \eqref{eq:Z_j_in_equiv_reduction} into \eqref{eq:Z_j_bar_in_equiv_reduction}, one can easily check the following. Equations 
\eqref{eq:Z_j_bar_in_equiv_reduction} and  \eqref{eq:Z_j_bar_in_equiv_reduction_smaller_problem}, as  linear equations on $Z$ and $\tilde{Z}$ respectively,
have the same coefficients if \eqref{eq:reduction_equations} holds. Hence, if \eqref{eq:reduction_equations} holds, the solution to the linear system associated to the optimality conditions in both problem gives the same optimal value for all variables ancestral to $\bar{j}$ and including $\bar{j}$.
\end{proof}
\begin{proof}[Proof of Theorem \ref{th:complexity_of_rec_reduce}]

Although $T$ changes during the execution of the algorithm, in the proof we let $T = (r,\V,\E)$ be the tree, passed to the algorithm at the zeroth level of the recursion. Recall that $|\V|= q$ and $\E = q-1$.

{\bf{Correctness: }} The correctness of the algorithm follows directly from Lemmas \ref{th:pruning}, \ref{th:star_problem}, and \ref{th:tree_reduction} and the explanation following these lemmas.

{\bf{Run-time: }} It is convenient to think of the complexity of the algorithm by assuming that it is running on a machine with a single instruction pointer that jumps from line to line in Algorithm \ref{alg:compute_rates_rec}. With this in mind, for example, the recursive call in line \ref{alg:line:rec:6} simply makes the instruction pointer jump from line \ref{alg:line:rec:6} to line \ref{alg:line:rec:1}. The run-time of the algorithm is bounded by the
sum of the time spent in each line in Algorithm \ref{alg:compute_rates_rec}, throughout its entire execution. Each basic step costs one unit of time.
Each node in $\V$ is only chosen as $j$ at most once, throughout the entire execution of the algorithm. Hence, line \ref{alg:line:rec:1} is executed at most $q$ times, and thus any line is executed at most $q$ times, at most once for each possible choice for $j$.

Assuming that we have $T.b$ updated, $j$ in line \ref{alg:line:rec:1} can be executed in $\mathcal{O}(1)$ time, by reading the first element of the linked list $T.b$. Lines \ref{alg:line:rec:2} and \ref{alg:line:rec:6} also take $\mathcal{O}(1)$ time. Here, we are thinking of the cost of line \ref{alg:line:rec:6} as simply the cost to make the instruction pointer jump from line \ref{alg:line:rec:1} to line \ref{alg:line:rec:6}, not the cost to 
fully completing the call to \emph{ComputeRatesRec} on the modified problem.
The  modification made to the $(T,\B,\alpha,\beta,\gamma)$-problem by lines \ref{alg:line:rec:5} and \ref{alg:line:rec:7}, is related to the addition, or removal, of at most $\text{degree}(j)$ nodes, where $\text{degree}(j)$ is the degree of $j$ in $T$. Hence, they can be executed in $\mathcal{O}(\text{degree}(j))$ steps. Finally, lines \ref{alg:line:rec:3} and \ref{alg:line:rec:8} require solving a star-shaped problem with $\mathcal{O}(\text{degree}(j))$ variables, and thus take $\mathcal{O}(\text{degree}{j})$, which can be observed by inspecting \eqref{eq:star_solution}.

Therefore, the run-time of the algorithm is bounded by $\mathcal{O}(\sum_j \text{degree}(j)) = \mathcal{O}(q)$.

To see that it is not expensive to keep $T$ updated, notice that, if $T$ changes, then either $T.b$ loses $j$ (line \ref{alg:line:rec:5}) or  has $j$ reinserted (line \ref{alg:line:rec:7}), both of which can be done in $\mathcal{O}(1)$ steps. Hence, we can keep $T.b$ updated with only $\mathcal{O}(1)$ effort each time we run
line \ref{alg:line:rec:5} and line \ref{alg:line:rec:7}.
Throughout the execution of the algorithm, the tree $T$ either shrinks by loosing nodes that are children of the same parent (line \ref{alg:line:rec:5}), or $T$ grows by regaining nodes that are all siblings (line \ref{alg:line:rec:7}). Hence, the linked list $T.a$ can be kept updated with only $\mathcal{O}(1)$ effort each time we run line \ref{alg:line:rec:5} and line \ref{alg:line:rec:7}.
Across the whole execution of the algorithm, $T.a$ and $T.b$ can be kept updated with $\mathcal{O}(\sum_j \text{degree}(j)) = \mathcal{O}(q)$ effort.

{\bf{Memory: }} All the variables with a size that depend on $q$ are passed by reference in each call of \emph{ComputeRatesRec}, namely, $Y$, $T$, $\B$, $\alpha$, $\beta$ and $\gamma$. Hence, we only need to allocate memory for them once, at the zeroth level of the recursion. All these variables take $\mathcal{O}(q)$ memory to store.
\end{proof}

\begin{proof}[Proof of Lemma \ref{th:rates_only}]
From Definition \ref{def:sub_problem_T_B_A_B_G}, we know that the $(T,\B,{\bf 1},-N,{\bf 1})$-problem
and the $(T,\B)$-problem are the same.
Hence, it is enough to prove that the solutions of (i) any $(T,\B,\alpha,\beta,\gamma)$-problem and of (ii) the 
$(T,\B,\alpha,0,\gamma)$-problem change at the same rate as a function of $t$.

We have already seen that the $(T,\B,\alpha,\beta,\gamma)$-problem can be solved by recursively invoking Lemma \ref{th:tree_reduction} until we arrive at problems that are small enough to be solved via Lemma \ref{th:star_problem}. 

We now make two observations. First, while recursing, Lemma \ref{th:tree_reduction} always transform a $(\tilde{T},\tilde{\B},\tilde{\alpha},\tilde{\beta},\tilde{\gamma})$-problem into a smaller problem $(\tilde{\tilde{T}},\tilde{\tilde{\B}},\tilde{\tilde{\alpha}},\tilde{\tilde{\beta}},\tilde{\tilde{\gamma}})$-problem where, by \eqref{eq:reduction_equations}, $\tilde{\tilde{\gamma}}$ and $\tilde{\tilde{\alpha}}$ only depend on $\tilde{\alpha}$ and $\tilde{\gamma}$ but not on $\tilde{\beta}$. 

Second, while recursing, and each time Lemma \ref{th:star_problem} is invoked to compute an explicit value for some component of the solution via solving some star-shaped $(\tilde{\tilde{T}},\tilde{\tilde{\B}},\tilde{\tilde{\alpha}},\tilde{\tilde{\beta}},\tilde{\tilde{\gamma}})$-problem, the rate of change of this component with $t$, is a function of $\tilde{\tilde{\alpha}}$ and $\tilde{\tilde{\gamma}}$ only. We can see this from \eqref{eq:star_solution}.

Hence, the rate of change with $t$ of the solution of the $(T,\B,\alpha,\beta,\gamma)$-problem does not depend on $\beta$. So we can assume $\beta = 0$.
\end{proof}
\begin{proof}[Proof of Theorem \ref{th:compute_rates_main_theorem}]

{\bf{Correctness: }} The correctness of Algorithm \ref{alg:compute_rates} follows from the correctness of Algorithm \ref{alg:compute_rates_rec}. 

{\bf{Run-time and memory: }}
We can prune each $T_w$ in $\mathcal{O}(|T_w|)$ steps and $\mathcal{O}(1)$ memory using DFS. In particular, once we reach a leaf of $T_w$ that is free, i.e., not in $\B_w$, and as DFS travels back up the tree, we can prune from $T_w$ all the nodes that are free.
By Theorem \ref{th:complexity_of_rec_reduce}, the number of steps and memory needed to completely finish  line \ref{alg:line:rates:4} is $\mathcal{O}(|T_w|)$. The same is true to complete line \ref{alg:line:rates:5}.
Hence, the number of steps and memory required to execute the for-loop is $\mathcal{O}(\sum_w |T_w|)=\mathcal{O}(|T|)=\mathcal{O}(q)$.
Finally, by Theorem \ref{th:computing_ddL_from_dZ}, $\LL''$ can be computed from $Z'^*$ in $\mathcal{O}(q)$ steps using $\mathcal{O}(1)$ memory.
\end{proof}

\section{Details of the ADMM and the PGD algorithms in Section \ref{sec:num_res}} \label{app:sec:details_of_ADMM_and_PDG}

Here we explain the details of our implementations of the Alternating Direction Method of Multipliers (ADMM) and the Projected Gradient Descent (PGD) methods, applied to our problem.

\subsection{ADMM}

\subsubsection{ADMM for the primal problem}

We start by putting our initial optimization problem (\ref{eq:projection_onto_PPM}) into the following equivalent form:
\begin{equation}
\min_{M \in \mathbb{R}^{q}} \{f(M) = \frac{1}{2} \|F - U M \|^2\} + g(M),
\end{equation}
where $g(M)$ is the indicator function imposing the constraints on $M$:
\begin{equation}
g(M) := 
\begin{cases}
0, \qquad M \geq 0,  M^\top \ones = \ones, \\
+\infty, \ ~~\text{otherwise}.
\end{cases}
\end{equation}

In this formulation, our target function is a sum of two terms. We now proceed with the standard ADMM procedure, utilizing the splitting $f$, $g$. Our ADMM scheme iterates on the following variables $M, M_1, M_2, u_1, u_2, \in \mathbb{R}^q$. $M_1$ and $M_2$ are primal variables, $M$ is a consensus variable, and $u_1$ and $u_2$ are dual variables.
It has tunning parameters $\alpha,\rho \in \mathbb{R}$.

First, we evaluate the proximal map associated with the first term
\begin{equation}
M_1 \gets \argmin_{S \in \mathbb{R}^{q}} \frac{1}{2} \|F - U S \|^2 + \frac{\rho}{2} \|S - M + u_1 \|^2,
\end{equation}
where $S$ is a dummy variable. This map can be evaluated in closed form,
\begin{equation}
M_1 = (\rho I + U^\top U)^{-1} (\rho M - \rho u_1 + U^\top F). 
\end{equation}

Second, we evaluate the proximal map associated with the second term
\begin{equation}
M_2 \gets \argmin_{S \in \mathbb{R}^{q}} g(S) + \frac{\rho}{2} \|S - M + u_2 \|^2,
\end{equation}
where $S$ is again a dummy variable. This map is precisely the projection onto the simplex, which has been extensively studied in the literature; there are many fast algorithms that solve this problem exactly. We implemented the algorithm proposed in \cite{condat2016fast}.

Lastly, we perform the rest of the standard ADMM updates:
\begin{equation}
\begin{aligned}
M &\gets \frac{1}{2}(M_1 + u_1 +  M_2 + u_2), \\
u_1 &\gets u_1 + \alpha (M_1 - M), \\
u_2 &\gets u_2 + \alpha (M_2 - M). \\
\end{aligned}
\end{equation}
We repeat the above steps until a satisfactory precision is reached, and read off the final solution from the variable $M$. 

\subsubsection{ADMM for the dual problem}

We now apply ADMM to the dual problem (\ref{eq:dual_higher}). We start by incorporating the constraints into the target function to rewrite (\ref{eq:dual_higher}) as
\begin{equation}
\min_{Z, t} \{f(t) = t\} + \{h(Z) = \frac{1}{2} \|(U^\top)^{-1} Z\|^2 \} + g(t, Z),
\end{equation}
where 
\begin{equation}
g(t, Z) := 
\begin{cases}
0, \qquad t \ones - Z \geq N, \\
+\infty, \ ~~\text{otherwise},
\end{cases}
\end{equation}
is the indicator function imposing the constraints on $t, Z$. ADMM now splits the problem into three parts, each associated to one of the functions $f,g$ and $h$.

Our ADMM scheme will iterate on the following variables $Z, X_Z, X_{gZ}, u_Z, u_{gZ} \in \mathbb{R}^q$, and $t, X_t, X_{gt}, u_t, u_{gt} \in \mathbb{R}$. The variables $X_Z, X_{gZ},X_t, X_{gt}$ are primal variables, $t,Z$ are consensus variables, and $u_Z, u_{gZ},u_t, u_{gt}$ are dual variables.
It has tunning parameters $\alpha,\rho \in \mathbb{R}$.

First, we evaluate the proximal map for the first term
\begin{equation}
X_Z \gets \argmin_{S \in \mathbb{R}^q} \frac{1}{2} \|(U^\top)^{-1} S\|^2 + \frac{\rho}{2} \|S - Z + u_Z\|^2,
\end{equation}
where $S$ is a dummy variable. This map can be evaluated using an closed form formula:
\begin{equation}
X_Z = (\rho I + U^{-1} (U^{-1})^{\top})^{-1} \rho (Z - u_Z).
\end{equation}

Next, we evaluate the proximal map for the second term
\begin{equation}
X_t \gets \argmin_{S \in \mathbb{R}} S + \frac{\rho}{2} (S - t + u_t)^2,
\end{equation}
where $S$ is a dummy variable. 
Again, this can be solved straightforwardly:
\begin{equation}
X_t = \frac{\rho t - \rho u_t - 1}{\rho}.
\end{equation}
We then evaluate the proximal map for the third term, which involves the constraints
\begin{align}
& (X_{gZ}, X_{gt}) \gets \argmin_{S \in \mathbb{R}^q, S_t \in \mathbb{R}} g(S, S_t)
+ \frac{\rho}{2} \|(S, S_t)\nonumber\\
&  - (Z - u_{gZ}, t - u_{gt})\|^2,
\end{align}
where $S, S_t$ are dummy variables. This problem is a projection onto the polyhedron defined by the constraints, $t \ones - Z \geq N$, in $\mathbb{R}^{q+1}$. We developed an algorithm that solves this problem exactly in $\mathcal{O}(q \log q)$ steps. This is discussed in Section \ref{projection-onto-polyhedron}.

What is left to be done is the following part of the ADMM:
\begin{equation}
\begin{aligned}
& Z \gets \frac{1}{2}(X_Z + u_Z +  X_{gZ} + u_{gZ}), \\
& u_Z \gets u_Z + \alpha (X_Z - Z), \\
& u_{gZ} \gets u_{gZ} + \alpha (X_{gZ} - Z), \\
& t \gets \frac{1}{2}(X_t + u_t +  X_{gt} + u_{gt}), \\
& u_t \gets u_t + \alpha (X_t - t), \\
& u_{gt} \gets u_{gt} + \alpha (X_{gt} - t). \\
\end{aligned}
\end{equation}
We repeat the above steps until a satisfactory precision is reached, and read off the final solution from the variables $t$ and $Z$.

\subsection{PGD}

\subsubsection{PGD for the primal problem}

Implementing PGD is rather straightforward. For the initial problem (\ref{eq:projection_onto_PPM}), we simply do the following update:
\begin{equation}
M \gets \text{Proj-onto-Simplex}(M + \alpha U^\top (F - U M)),
\end{equation}
where Proj-onto-Simplex() refers to projection onto the simplex, for which we implemented the algorithm proposed in \cite{condat2016fast}. $\alpha \in \mathbb{R}$ is the step size, a tuning parameter. We perform this update repeatedly until a satisfactory precision is reached.

\subsubsection{PGD for the dual problem}

For the dual problem (\ref{eq:dual_higher}) , the updates we need are
\begin{equation}
\begin{split}
Z &\gets Z - \alpha U^{-1}(U^{-1})^\top Z, \\
t &\gets t - \alpha, \\
(Z, t) &\gets \text{Proj-onto-Polyhedron}((Z, t)),
\end{split}
\end{equation}
where Proj-onto-Polyhedron() refers to projection onto the polyhedron defined by $t \ones - Z \geq N$ in $\mathbb{R}^{q+1}$, while $\alpha \in \mathbb{R}$ is the step size. This is explicitly explained in \ref{projection-onto-polyhedron}. Again, we perform these updates repeatedly until a satisfactory precision is reached, and tune the parameters to achieve the best possible performance.
\subsection{Projection onto the polyhedron $t \ones - Z \geq N$} 
\label{projection-onto-polyhedron}
We would like to solve the following optimization problem:
\begin{gather} \label{proj_poly}
\argmin_{Z \in \mathbb{R}^q, t \in \mathbb{R}} \frac{1}{2} \|(Z, t) - (A, B)\|^2, \\
\text{subject to }~~t \ones - Z \geq N,
\end{gather}
which is the problem of projection onto the polyhedron $t \ones - Z \geq N$ in $\mathbb{R}^{q+1}$.
The Lagrangian of this optimization problem is
\begin{equation}
\mathcal{L} = \frac{1}{2} \|(Z, t) - (A,B)\|^2 + \lambda^\top (Z + N - t\ones),
\end{equation}
where $\lambda \in \mathbb{R}^q$ is the Lagrange multiplier. 
We solve problem \eqref{proj_poly} by solving the dual problem $\max_{\lambda \geq 0} \min_{Z,t}\mathcal{L}$.

We first solve the minimization over variables $Z$ and $t$. It is straightforward to find the closed form solutions:
\begin{equation} \label{proj-poly-var}
%\begin{split}
Z^* = A - \lambda, \qquad
t^* = B + \ones^\top \lambda.
%\end{split}
\end{equation}
Using these expressions, we can rewrite the Lagrangian as
\begin{equation}
\mathcal{L} = -\frac{1}{2}\lambda^\top (I + \ones\ones^\top) \lambda + R^\top \lambda, 
\end{equation}
where $R = A+ N - B \ones$.

Now our goal becomes solving the following optimization problem:
\begin{gather} \label{proj_dual}
\argmin \frac{1}{2} \lambda^\top (I + \ones\ones^\top) \lambda - R^\top \lambda, \\
\text{subject to }~~
\lambda \geq 0.
\end{gather}
The KKT conditions for (\ref{proj_dual}) are
\begin{equation}
\lambda_i + \ones^\top \lambda - R_i - s_i = 0, \ \ \lambda_i \geq 0, \ \ s_i \geq 0, \  \lambda_i s_i = 0, \\ i = 1,..,q,
\end{equation}
where $s_i$ are Lagrange multipliers associated with the constraint $\lambda \geq 0$.

We proceed with sorting the vector $R$ first, and maintain a map $f: \{1, 2, ..., q\} \rightarrow \{1, 2, ..., q\}$ that maps the sorted indices back to the unsorted indices of $R$. Let us call the sorted $R$ by $\tilde{R}$. Then, from the above KKT conditions, it is straightforward to derive the following expression for $\lambda_i$:
\begin{equation} \label{lambda_i}
\lambda_i = 
\begin{cases}
&\tilde{R}_i - \ones^\top \lambda, \ i \geq \tau, \\
&0, \ i < \tau,
\end{cases}
\ \ i = 1, 2, ..., q
\end{equation}
where
\begin{equation} \label{eq:def_of_tau}
\tau = \min \{i \ | \ \tilde{R}_i - \ones^\top \lambda \geq 0\}.
\end{equation}
Then it follows that
\begin{equation} \label{1_lambda}
\begin{aligned}
\ones^\top \lambda = \sum_{i=\tau}^{q} (\tilde{R}_i - \ones^\top \lambda) = \frac{1}{2 + q - \tau} \sum_{i=\tau}^{q} \tilde{R}_i,
\end{aligned}
\end{equation}
and hence we have that
\begin{equation}
c(\tau):= \tilde{R}_\tau - {\bf 1}^\top \lambda = \tilde{R}_\tau -  \frac{1}{2 + q - i} \sum_{j=\tau}^{q} \tilde{R}_j.
\end{equation}
According to \eqref{eq:def_of_tau}, to find $\tau$, we only need to find the smallest value of $i$ that makes $c(i)$ non negative. That is, $\tau = \min \{ i \ | \ c(i) \geq 0 \}$.

Therefore, by sorting the components of $R$ from small to large, and checking $c(i)$ for each component, from large $i$ to small $i$, we can obtain the desired index $\tau$. Combining equations (\ref{lambda_i}) and (\ref{1_lambda}) with $\tau$, we find a solution that equals $\lambda^*$,  the solution to problem (\ref{proj_dual}), apart from a permutation of its components.
We then use our index map $f$ to undo the sorting of the components introduced by sorting $R$.

Finally, by plugging $\lambda^*$ back into equation (\ref{proj-poly-var}), we obtain the desired solution to our problem (\ref{proj_poly}). The whole projection procedure can be done in $\mathcal{O}(q \log q)$, the slowest step being the sorting of $R$.

\section{More results using our algorithm}\label{sec:app:more_results}

In this section, we use our fast projection algorithm to infer phylogenetic trees from frequency of mutation data.

The idea is simple. 
We scan all possible trees, and, for each tree $T$, we project $\hat{F}$ into a PPM for this $T$ using our fast projection algorithm. This gives us a projected $F$ and $M$ such that $F = UM$, the columns of $M$ are in the probability simplex, and $\|\hat{F} - F\|$ is small. Then, we return the tree whose projection yields the smallest $\|\hat{F} - F\|$. Since all of these projections can be done in parallel, we assign the projection for different subsets of the set of all possible trees to different GPU cores. Since we are performing an exhaustive search over all possible trees, we can only infer small trees. As such, when dealing with real-size data, similar to several existing tools, we first cluster the rows of $\hat{F}$, and produce an ``effective'' $\hat{F}$ with a small number of rows. We infer a tree on this reduced input. Each node in our tree is thus associated with multiple mutated positions in the genome, and multiple mutants, depending on the clustering. We cluster the rows of $\hat{F}$ using $k$-means, just like in \cite{malikic2015clonality}. We decide on the numbers of clusters, and hence tree size, based on the same BIC procedure as in \cite{malikic2015clonality}. It is possible that other pre-clustering, and tree-size-selection strategies, yield better results. We call the resulting tool EXACT. 

We note that it is not our goal to show that the PPM is adequate to extract phylogenetic trees from data. This adequacy, and its limits, are well documented in well-cited biology papers. Indeed, several papers provide open-source tools based on the PPM, and show their tools' good performance on
data containing the frequencies of mutation per position in different samples, $\hat{F}$ in our paper. A few tools are PhyloSub \cite{jiao2014inferring}, AncesTree \cite{el2015reconstruction}, CITUP \cite{malikic2015clonality}, PhyloWGS \cite{deshwar2015phylowgs}, Canopy \cite{jiang2016assessing}, SPRUCE \cite{el2016inferring}, rec-BTP \cite{hajirasouliha2014combinatorial}, and LICHeE \cite{popic2015fast}. These papers also discuss the limitations of the PPM regarding inferring evolutionary trees, and others propose extensions to the PPM to capture more complex phenomena, see e.g., \cite{bonizzoni2014and}. 

It is important to further distinguish the focus of our paper from the focus of the papers cited in the paragraph above. In this paper, we start from the fact that the PPM is already being used to infer trees from $\hat{F}$, and with substantiated success. However, all of the existing methods are heuristics, leaving room for improvement. { We identify one subproblem that, if solved very fast, allows us to do { exact PPM-based tree inference} for problems of relevant biological sizes. It is this subproblem, a projection problem in Eq. (3), that is our focus.
We introduce { the first non-iterative algorithm} to solve this projection problem, and show that it is 74$\times$ faster than different optimally-tuned iterative methods.
We are also { the first} to show that a full-exact-enumeration approach to inferring $U$ and $M$ from $\hat{F}$ is possible, in our case, using a GPU and our algorithm to compute and compare the cost of all the possible trees that might explain the data $\hat{F}$.}
EXACT often outperforms the above tools, {none of which does exact inference}. 
Our paper is not about EXACT, whose  development challenges and significance for biology go beyond solving our projection problem, and which is the focus of our future work. 

Despite this difference in purpose, in this section we compare the performance of inferring trees from a full exact search over the space of all possible PPM models with the performance of a few existing algorithms.
In Figure \ref{fig:app:more_results},  we compare EXACT, PhyloWGS, CITUP and AncesTree on recovering the correct ancestry relations on  biological datasets also used by [3]. A total of $30$ { different datasets} \cite{ancestreedata}, i.e., $\hat{F}$, were tested. We use the default parameters in all of the algorithms tested. 

% \begin{wrapfigure}[6]{r}{0.72\textwidth}
\begin{figure*}[h!]
{\centering 
\includegraphics[trim={0.1cm 0 0 0},clip,width=\textwidth]{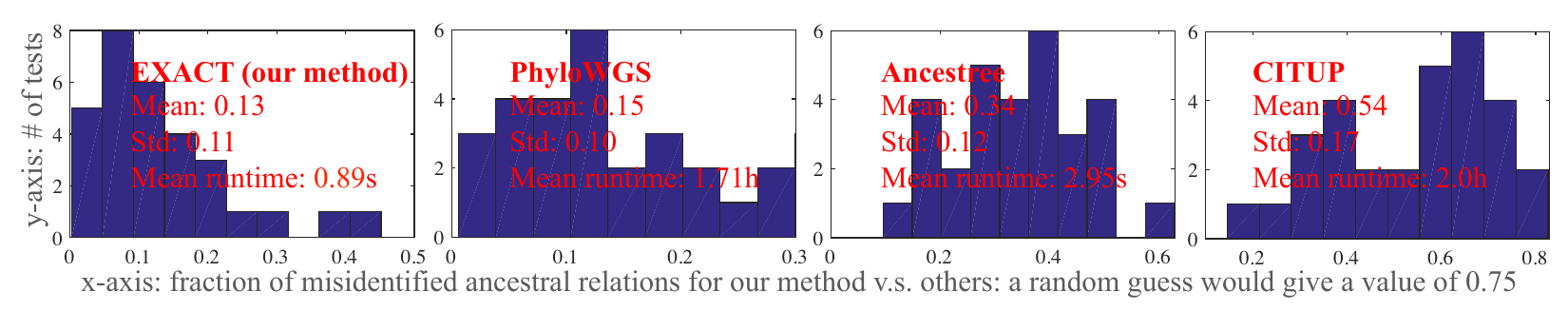}
}
%\end{wrapfigure}
\caption{\small Comparison of different phylogenetic tree inference algorithms.}\label{fig:app:more_results}
\end{figure*}

In each test, and for every pair of mutations $i$ and $j$, we use the tree output by each tool to determine if (a) $i$ is an ancestor of $j$ or if $j$ is an ancestor of $i$, (b) if $i$ and $j$ are in the same node, (c) if either $i$ or $j$ are missing in the tree, or,  otherwise, (d) if $i$ and $j$ are {incomparable}. We give these four possible ancestral relations, the following names: \emph{ancestral}, \emph{clustered},  \emph{missing}, and \emph{incomparable}. A random guess correctly identifies $25$\% of the ancestral categories, on average. If the fraction of misidentified relations is $0$, the output tree equals the ground-truth tree. All methods do better than random guesses. 

For example, in Figure \ref{fig:app:example_of_reconstruction}, according to EXACT, mutation $63$, at the root, is an ancestor of mutation $57$, at node $3$. However, according to the ground truth, in Figure \ref{fig:app:example_of_reconstruction}, they belong to the same node. So, as far as comparing $63$ with  $57$ goes, EXACT makes a mistake. As another example, according tp EXACT, mutations $91$ and $55$ are incomparable, while according to the ground truth, $91$ is a descendent of $55$. Hence, as far as comparing $91$ with $55$ goes, EXACT makes another mistake.
The fraction of errors, per ancestral relation error type, that each of these tools makes is: EXACT = $\{23\%,10\%,0\%,13\%\}$; PhyloWGS = $\{3\%,2\%,0\%,1\%\}$; AncesTree = $\{54\%,16\%,95\%,25\% \}$; CITUP = $\{27\%,13\%,0\%,21\%\}$.

In our experiments, EXACT performs, on average, better than the other three methods. PhyloWGS  performs close to EXACT, however, it has a much longer run time. %
Although AncesTree does fairly well in terms of accuracy, we observe that it often returns trees with the same topology, a star-shaped tree. The other methods, produce trees whose topology seems to be more strongly linked to the input data. Finally, AncesTree's inferred tree does not cover all of the existing mutations. This behaviour is expected, as, by construction, AncesTree tries to find the largest tree that can be explained with the PPM. See Figure \ref{fig:app:example_of_reconstruction}, and Figure \ref{fig:app:ground_truth_tree}, for an example of the output produced by different algorithms, and the corresponding ground truth.

\begin{figure*}[h!]
{\centering 
\includegraphics[trim={0.0cm 0 0 0},clip,width=\textwidth]{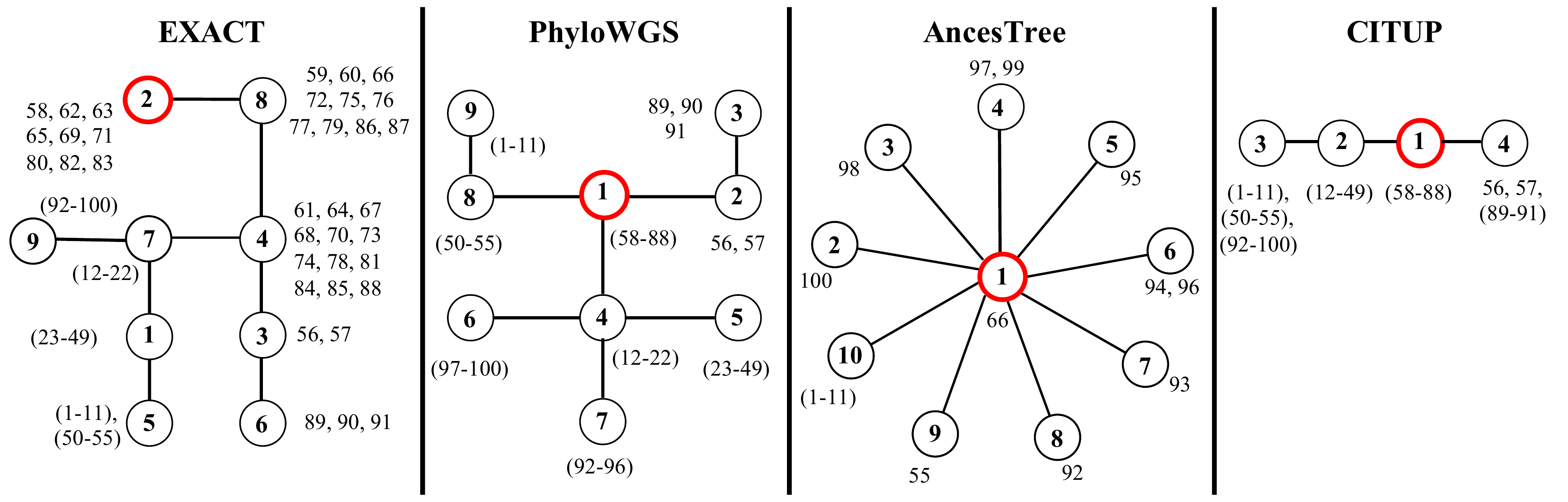}
}
%\end{wrapfigure}
\caption{ \small Tree reconstructed by different algorithms for the first file in the folder \cite{ancestreedata}. AncesTree often outputs star-shaped trees. The small numbers listed next to each node represent mutations. Mutations indexed by the same number in different trees are the same real mutation. The root of each tree is circled in thick red. Nodes are labeled by numbers, and these labels are assigned automatically by each tool. Labels of different trees are incomparable.  }\label{fig:app:example_of_reconstruction}
\end{figure*}

\begin{figure}[h!]
{\begin{center}
\includegraphics[trim={0.0cm 0 0 0},clip,width=3.5cm]{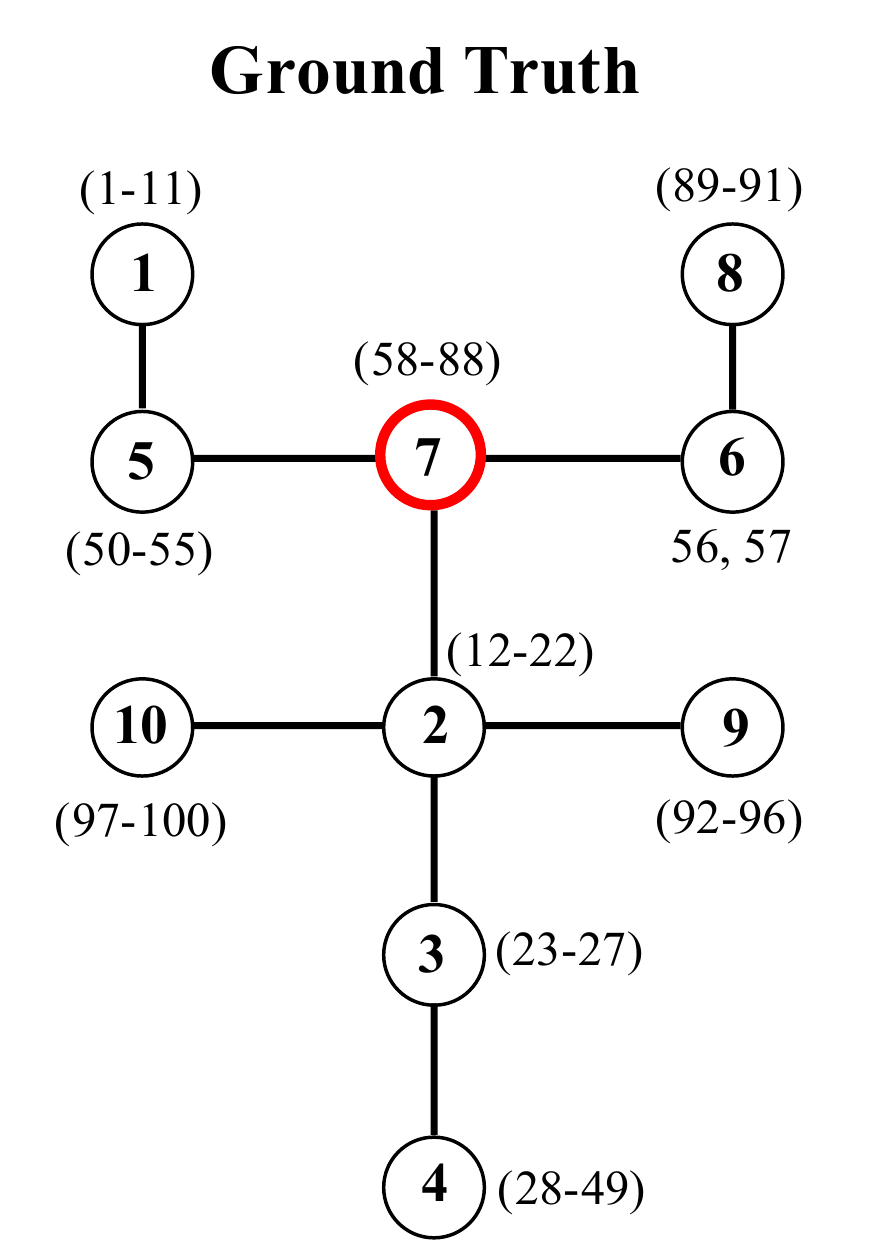}
\end{center}
}
%\end{wrapfigure}
\caption{ \small Ground truth tree for the input file that generated Figure \ref{fig:app:example_of_reconstruction}. The small numbers listed next to each node represent mutations. Mutations indexed by the same number in different trees are the same real mutation. The root of each tree is circled in thick red. Nodes are labeled by numbers, and these labels are assigned automatically by each tool. Labels of different trees are incomparable.  }\label{fig:app:ground_truth_tree}
\end{figure}

We end this section by discussing a few extra properties that distinguished an approach like EXACT from the existing tools.
Because our algorithm's speed allows a complete enumeration of all of the trees, EXACT { has  two unique properties}. First, EXACT  can exactly solve 
\begin{equation}\label{eq:app:gen_likelihood}
\min_{U \in \mathcal{U}}  \mathcal{J}(\mathcal{C}(U)) + \mathcal{Q}(U),
\end{equation}
where $U$ encodes ancestral relations, $\mathcal{C}(U)$ is the fitness cost as defined in our paper,  $\mathcal{J}$ is an arbitrary, fast-to-compute, 1D scaling function, and $\mathcal{Q}(U)$ is an arbitrary, fast-to-compute, tree-topology penalty function. No other tool has this flexibility. Second, EXACT can find the $k$ trees with the smallest objective value in \eqref{eq:app:gen_likelihood}. A few existing tools can output multiple trees, but only when these all have the same ``heuristically-optimal'' objective value. This feature is very important because, given that the input data is noisy, and the number of samples is often small, it allows, e.g., one to give a confidence score for the ancestry relations in the output tree. Furthermore, experiments show that the ground-truth tree can often be found among these $k$ best trees. Hence, using other biological principles, the ground-truth tree can often be identified from this set. Outputting just ``heuristically-optimal'' trees prevents this finding.
%We hope the above discussions address the reviewer's concern, and we are looking forward to his/her support of our manuscript.

\end{document}